\documentclass[a4paper, 11pt]{article}

\usepackage[
	backend=biber,
	maxbibnames=99,
    giveninits=true
]{biblatex}
\usepackage{algorithm,algpseudocode}
\usepackage{appendix}
\usepackage{authblk}
\usepackage{booktabs}			
\usepackage{enumitem}			
\usepackage{fancyvrb}			
\usepackage{float}			
\usepackage[T1]{fontenc}		
\usepackage{hyperref}			
\usepackage{manyfoot}
\usepackage{microtype}			
\usepackage[normalem]{ulem}		
\usepackage{url}			
\usepackage{amsfonts,amsmath,amssymb,amsthm}
\usepackage{bm,bbm}			
\usepackage{mathrsfs}			
\usepackage{mathtools}			
\usepackage[super]{nth}			
\usepackage{oubraces}			
\usepackage{stmaryrd}			

\usepackage{graphicx}
\usepackage[caption=false]{subfig}
\usepackage{tikz}
\usepackage{tkz-graph}
\usetikzlibrary{external}
\usetikzlibrary{cd, matrix}
\usetikzlibrary{decorations.pathreplacing}

\def\dfn#1{{\boldmath{\textbf{#1}}}}

\newcommand{\abs}[1]{\left|#1\right|}

\DeclareMathOperator{\dist}{d}

\DeclareMathOperator{\supp}{supp}
\DeclareMathOperator{\wt}{wt}

\newcommand{\F}{\mathbb{F}}
\newcommand{\K}{\mathbb{K}}
\newcommand{\N}{\mathbb{N}}
\newcommand{\R}{\mathbb{R}}

\newcommand{\CurA}{\mathcal{A}}
\newcommand{\CurB}{\mathcal{B}}
\newcommand{\CurC}{\mathcal{C}}
\newcommand{\CurD}{\mathcal{D}}
\newcommand{\CurP}{\mathcal{P}}
\newcommand{\CurR}{\mathcal{R}}

\newcommand{\Cpro}{\otimes}

\theoremstyle{plain}%
\newtheorem{theorem}{Theorem}
\newtheorem{corollary}[theorem]{Corollary}%
\newtheorem{lemma}[theorem]{Lemma}%
\newtheorem{proposition}[theorem]{Proposition}%

\theoremstyle{remark}%
\newtheorem{example}[theorem]{Example}%
\newtheorem{remark}[theorem]{Remark}%

\theoremstyle{definition}%
\newtheorem{definition}[theorem]{Definition}%

\begin{document}

\title{q-ary Sequential Locally Recoverable Codes from the Product Construction}

\author[1,2]{Akram Baghban \href{mailto:a.baghban@sci.basu.ac.ir}{a.baghban@sci.basu.ac.ir}}
\author[2]{Marc Newman \href{mailto:marc.newman@unisg.ch}{marc.newman@unisg.ch}}
\author[2]{Anna-Lena Horlemann \href{mailto:anna-lena.horlemann@unisg.ch}{anna-lena.horlemann@unisg.ch}}
\author[1]{Mehdi Ghiyasvand \href{mailto:mghiyasvand@basu.ac.ir}{mghiyasvand@basu.ac.ir}}

\affil[1]{%
	Department of Mathematics,
	Bu-Ali Sina University,
	Mostafa Ahmadi Roshan, 
	Hamedan, 6516738695, Iran
}

\affil[2]{%
	Institute of Computer Science,
	University of St.Gallen,
	Torstrasse 25, 9000, St.\ Gallen, Switzerland
}

\date{}

\maketitle{}

\begin{abstract}
	\noindent{}This work focuses on sequential locally recoverable codes (SLRCs), a special family of locally repairable codes, capable of correcting multiple code symbol erasures, which are commonly used for distributed storage systems. 
	First, we construct an extended $q$-ary family of non-binary SLRCs using product codes with a novel maximum number of recoverable erasures $t$ and a minimal repair alternativity $A$. 
	Second, we study how MDS and BCH codes can be used to construct $q$-ary SLRCs.
	Finally, we compare our codes to other LRCs.

	\vspace{\baselineskip}

	\noindent{}
\end{abstract}


\section{Introduction}

In a distributed storage system (DSS) data is stored in a large, distributed network of nodes or storage units.
To protect the system against node failure, current DSS systems utilize different coding schemes---specifically locally recoverable codes (LRCs)---which have recently been the subject of study.
The purpose of an LRC is to reduce the total number of required nodes employed for node repair.
Generally, an LRC with locality $r$ is a linear code over a finite field where the value of every code symbol is recoverable by accessing at most $r$ other symbols.

In such systems, it is common for multiple storage nodes to fail simultaneously. To address this, two primary approaches have been developed for locally recoverable codes (LRCs) capable of handling multiple erasures: parallel recovery and sequential recovery. In the parallel approach, each erased code symbol must be recovered independently using only the symbols that remain uncorrupted. In contrast, the sequential approach allows the erased symbols to be ordered and recovered one at a time, where each symbol can be reconstructed using both the initially available symbols and those recovered earlier in the sequence.

This sequential strategy offers greater flexibility than its parallel counterpart, as the choice of recovery sets can adapt dynamically throughout the process. As a result, sequential recovery can potentially handle a larger number of erasures.
Nevertheless, codes that support parallel recovery are often favored in practice due to their ability to reconstruct multiple erased symbols simultaneously, thereby reducing overall recovery latency and enhancing performance in scenarios with high read demand.

The interested reader is referred to the following works on locally recoverable codes. 
In~\cite{GHH12} Gopalan et al.\ presents the interesting notion of locally recoverable codes (see also~\cite{OD11} and~\cite{PD14}) in which a code symbol that has been erased is able to be recovered using a small set of other code symbols.
The authors of~\cite{GHH12} discuss codes with local single erasure recovery possibilities (see~\cite{BKVRSK18},~\cite{HCL13},~\cite{KPLK14}, and~\cite{TBF16}).

The sequential recovery developed by Prakash et al.~\cite{PLBK19} includes several techniques for locally recovering multiple erasures (see also~\cite{BPK16-1} and~\cite{RMV15}).
The authors in~\cite{PLBK19} discuss codes which can recover sequentially following two-erasures (see~\cite{SCYCH18}).
Codes which sequentially recover three-erasures are covered in~\cite{BPK16-1},~\cite{SCYCH18}, and~\cite{SY15}.
Alternate methods related to local recovery from multiple erasures are presented in~\cite{BKVRSK18},~\cite{BKK19}, and~\cite{BPK16-2}.

This research introduces a $q$-ary family of non-binary SLRCs which are capable of recovering $t \geq 2$ multiple erasures having small locality $r$.
With our structure we establish new upper bounds on the total number of recoverable erasures and code rate which significantly improve the results from~\cite{BPK16-1},~\cite{SCYCH18}, and~\cite{WZL15} using product codes.
Additionally, from the extension of SLRC constructions in~\cite{BPK16-1} in the case of binary codes, weaknesses in specific conditions on $k$ and $t$ are eliminated and a broader range of code parameters are developed. 

The paper is structured as follows: in Section 2 we provide necessary preliminaries, in Section 3 we provide our general $q$-ary construction for SLRCs, in Section 4 we provide explicit constructions using products of MDS and BCH codes, and in Section 5 we compare our constructions to other known constructions of general LRCs.


\section{Preliminaries}\label{sec:prelim}

We fix some notation that will be used throughout.
For any subset $A$, $\abs{A}$ is the size of $A$.
A set $B$ is referred to as a $t$-subset of $A$ if $B \subseteq A$ and $\abs{B} = t$.
For any positive integer $n$, we denote by $[n]$ the set $\{1, 2, \ldots, n\}$. The set $\F_q$ is the finite field with $q$ elements. Furthermore, let $\supp(x)$ denote the support of a vector $x$.

\subsection{Linear Codes and Local Recovery}

\begin{definition}
	A $q$-ary \textbf{linear code} of length $n$ and dimension $k$, denoted as an ${[n, k]}_q$-code, is a dimension $k$ linear subspace $\CurC$ of the vector space $\F_q^n$.
    The elements of a code are referred to as its codewords.

	The \textbf{rate} of an ${[n, k]}_q$-linear code is defined to be $\frac{k}{n}$.
\end{definition}

Let $\CurC$ be an ${[n, k]}_q$-linear code and $S$ be a $k$-subset of $[n]$.
Then the elements of $S$ denote a \textbf{set of information symbols} of $\CurC$ if there exist $a_{ij}\in \F_q$, for $i\in [n], j\in S$, such that for all codewords $x = (x_1, x_2, \ldots, x_n) \in \CurC$, $x_i = \sum_{j \in S} {a_{ij} x_{j}}$.

\begin{definition}
	The \textbf{Hamming weight} of a vector $x$, denoted by $\wt(x)$, is defined as the number of non-zero entries of $x$, i.e.,
	\begin{align*}
		\wt(x) = \abs{\supp(x)}.
	\end{align*}
    The \textbf{Hamming distance} between two vectors $x$ and $y$ is the Hamming weight of their difference, i.e., $\dist(x, y) = \wt(x - y)$.
\end{definition}
\begin{definition}
	The \textbf{minimum distance} of a code $\CurC$, denoted by $\dist(\CurC)$, is defined as the minimal distance between any two distinct codewords in the code, i.e.,
    $$\dist(\CurC) = \min\{d(x,y) \mid x,y \in \CurC, x\neq y\} .$$
    In a linear code, this is equivalent to the minimal Hamming weight of any nonzero codeword in $\CurC$.
\end{definition}

\begin{lemma}\cite{MS77}\label{lemma:linind1}
	Let $\CurC$ be an ${[n, k]}_q$-linear code with generator matrix $G$.
	Then any $\dist(\CurC^\perp) - 1$ columns of $G$ are linearly independent.
\end{lemma}

Now we recall some definitions from~\cite{SCYCH18} and~\cite{SY15} related to (sequential) locally recoverable codes.

\begin{definition}
	Let $i \in [n]$ and $R_i \subseteq [n] \setminus \{i\}$.
	The subset $R_i$ is called an $(r, \CurC)$-\dfn{recovery set} of $i$ if $\abs{R_i} \leq r$ and there exist $a_{ij}\in \F_q$ such that  $x_i=\sum_{j\in R_i}{a_{ij}x_{j}}$ for all $(x_1, \ldots, x_n) \in \CurC$.

	We say a code $\CurC$ has \dfn{locality} $r$ if every symbol of $\CurC$ has a $(r, \CurC)$-recovery set.
	An ${[n, k]}_q$ code with locality $r$ is referred to as an $(n, k, r)$-\dfn{locally recoverable code (LRC)}.
\end{definition}

Many constructions of LRCs exist and bounds on the rate and a Singleton-like bound have been given for them.

\begin{theorem}\cite{TB14}\label{thm:optimal}
	Let $\CurC$ be an ${(n, k, r)}_q$ LRC with minimum distance $d$.
	Then
	\begin{align*}
		\frac{k}{n} \leq \frac{r}{r + 1},
		\qquad
		d \leq n - k - \left\lceil \frac{k}{r} \right\rceil + 2.
	\end{align*}
	An LRC with equality in the second bound is referred to as an \dfn{optimal LRC code}.
\end{theorem}

The optimality above is with respect to the minimum distance of a locally recoverable code. If we have disjoint recovery sets, then the information theoretic optimal codes are \textbf{partial maximum distance separable (PMDS) codes}.\footnote{The PMDS property always implies the Hamming distance optimality from Theorem~\ref{thm:optimal}.} 
%
%
In other words, a code is PMDS if the projection to any recovery set is an MDS code, and if additionally we can recover the maximal amount of global erasures, i.e., additional erasures anywhere in the code. 

\begin{definition}\cite{HN20}
	Let $r, m, t_1, \ldots, t_m \in \N$, set $n = \sum_{i = 1}^m (r + t_i)$, and let $\CurC \subseteq \F_q^n$ be the linear code given by the generator matrix $G = (B_1 | \cdots | B_m) \in \F_q^{k \times n}$ where each $B_i \in \F_q^{k \times (r + t_i)}$.
	Then, we call $\CurC$ an $[n, k, r; t_1, \ldots, t_m]$-\dfn{partial MDS (PMDS) code} if
	\begin{itemize}
		\item each $B_i$ is the generator matrix of an $[r + t_i, r]$ MDS code and
		\item for any combination of $\sum_{i = 1}^m t_i$ erasures with $t_i$ erasures in $i$-th block for each $i$, the remaining code (after puncturing the coordinates of the erasures) is an $[mr, k]$-MDS code.
	\end{itemize}
\end{definition}

For simplicity we assume $t_1 = \cdots = t_m = t'$ in the following.

\begin{proposition}
	A $[n, k, r; t', \ldots, t']$-PMDS code can correct at least $t = n - k - t' \lceil \tfrac{k}{r} \rceil$ erasures, out of which $t_i$ erasures can be recovered with locality $r$ in the $i$-th recovery set corresponding to $B_i$ above, for any $i\in \{1,\dots,m\}$. The remaining erasures need to be recovered globally in the whole code.
\end{proposition}

In this paper we are mainly interested in the case where we do not have disjoint recovery sets, and where sequential recovery is beneficial over parallel recovery. We explain these concepts in the following.

\begin{definition}
	Let $E$ be a $t$-subset of $[n]$.
	$\CurC$ is said to be $(E, r)$-\dfn{recoverable} if there exists a sequential indexing of the elements of $E$, say $(i_1, i_2, \ldots, i_{t})$, such that each $i_j \in E$ has an $(r, \CurC)$-recovery set $R_{i_j} \subseteq \overline{E} \cup \{i_1, i_2, \ldots, i_{j - 1}\}$ of size $\abs{R_{i_j}} \leq r$ where $\overline{E} = [n] \setminus E$.
	An ${[n,k]}_q$-linear code $\CurC$ is said to be an $(n, k, r, t)$-\dfn{sequential locally recoverable code (SLRC)} if for each $E \subseteq [n]$ of size $\abs{E} \leq t$, $\CurC$ is an $(E, r)$-recoverable code.
\end{definition}

Note that $t$ is the upper bound of the number of (sequential locally) recoverable failed nodes and that by definition we have $2 \leq \abs{R_i} \leq r \leq k < n$ where $R_i$ is a recovery set for some $(n, k, r, t)$-SLRC\@.

\begin{proposition}\cite{SCYCH18}\label{prop:SLRCequiv}
	$\CurC$ is an $(n, k, r, t)$-SLRC if and only if for any nonempty $E \subseteq [n]$ of size $\abs{E} \leq t$ there exists an $i \in E$ such that $i$ has an $(r, \CurC)$-recovery set $R_i \subseteq [n] \setminus E$.
\end{proposition}

Another interesting parameter of locally recoverable codes is its {repair alternativity}, which refers to the number of different recovery sets available to  recover lost data:

\begin{definition}\cite{PJHO13} 
	Let $\CurC$ be an $(n, k, r, t)$-SLRC\@ and for $i \in [n]$ let
	\begin{align*}
		\Omega_r(i) = \{v \in \CurC^\perp \mid i \in \supp(v), \wt(v) \leq r + 1\}.
	\end{align*}
	We say the \dfn{repair alternativity} of $i \in [n]$ is
	\begin{align*}
		a(i) = \abs{\{\supp(v) \mid v \in \Omega_r(i)\}}
	\end{align*}
	and that the \dfn{repair alternativity} of the code $\CurC$ is
	\begin{align*}
		a = \min_{i \in [n]} \{a(i)\}.
	\end{align*}
\end{definition}

In the setting above we call the elements of $\Omega_r(i)$ the \dfn{recovery vectors} for coordinate $i$ in the code $\CurC$.



\begin{remark}
    A concept similar to alternativity is \textbf{availability}, which is often used in the context of parallel recovery. However, in the setting of sequential local recovery, alternativity is more relevant. While availability quantifies the number of disjoint recovery sets for simultaneous repairs in parallel recovery, sequential recovery relies on the adaptive use of previously recovered symbols, making alternativity the appropriate measure for characterizing the flexibility in the recovery process.
\end{remark}

It is easy to see and well-known what the locality and alternativity of MDS codes are. For completeness we include a proof below.
\begin{proposition}
	Let $\CurC$ be an ${[n, k, d]}_q$-MDS code.
	Then any erased symbol in $\CurC$ can be recovered from any other $k$ symbols.
	In other words, an $[n, k, d]$-MDS code has locality $k$ and alternativity $\binom{n - 1}{k}$.
\end{proposition}

\begin{proof}
	Since all maximal minors of an MDS generator matrix are non-zero, any erased symbol can be recovered by solving a system of $k$ linear equations, corresponding to any $k$ non-erased symbols. Therefore, the locality is $k$. 
	Since we can use any $k$ (other symbols) for a fixed erased symbol, there are $\binom{n-1}{k}$ choices for these other symbols.
\end{proof}


\subsection{The Kronecker Product of Matrices and Codes}

As our code construction mainly relies on the Kronecker product of generator matrices of codes, we will define it and state its main properties in this subsection.

\begin{definition}
	Let $A$ be an $m_1 \times n_1$ matrix with entries
	\begin{align*}
		A = \begin{bmatrix}
			a_{1 1} & \cdots & a_{1 n_1} \\
			\vdots & \ddots & \vdots \\
			a_{m_1 1} & \cdots & a_{m_1 n_1}
		\end{bmatrix}
	\end{align*}
	and let $B$ be an $m_2 \times n_2$ matrix.
	Then the \dfn{Kronecker product} $A \otimes B$ is the $m_1 m_2 \times n_1 n_2$ matrix
	\begin{align*}
		A \otimes B = \begin{bmatrix}
			a_{1 1} B & \cdots & a_{1 n_1} B \\
			\vdots & \ddots & \vdots \\
			a_{m_1 1} B & \cdots & a_{m_1 n_1} B
		\end{bmatrix}.
	\end{align*}
\end{definition}


\begin{proposition}\cite{HJ94}
	Let $A$, $B$, $C$, and $D$ be matrices such that $AC$ and $BD$ are valid products.
	Then
	\begin{itemize}
		\item $(A \otimes B) (C \otimes D) = (AC) \otimes (BD)$,
		\item ${(A \otimes B)}^\top = A^\top \otimes B^\top$.
	\end{itemize}
\end{proposition}

\begin{lemma}\label{prop:tenvec}
	Let $\K$ be an arbitrary field and let $u \in \K^m$ and $v \in \K^n$ be vectors.
	Then if $i \in [mn]$, there exists a unique pair $(j, k) \in [m] \times [n]$ such that $i = (j - 1) n + k$ and $i \in \supp(u \otimes v)$ if and only if $j \in \supp(u)$ and $k \in \supp(v)$.
\end{lemma}

\begin{proof}
	The map 
	\begin{align*}
		f:
		[m] \times [n] &\to [mn] \\
		(j, k) &\mapsto (j - 1)n + k
	\end{align*}
	is a bijection.
	Furthermore, if $u = (u_1, \ldots, u_m)$ and $v = (v_1, \ldots, v_n)$, then given $i \in [mn]$ and $(j, k) = f^{-1}(i)$, the $i$th entry of the product $u \otimes v$ is equal to $u_j \cdot v_k$.
\end{proof}

\begin{definition}\cite{MS77} 
	Let $\CurA$ and $\CurB$ be $[n_1, k_1]$ and $[n_2, k_2]$ linear codes with generator matrices $G_A$ and $G_B$, respectively.
	Then the \dfn{product code} $\CurA \Cpro \CurB$ is the $[n_1 n_2, k_1 k_2]$ linear code with generator matrix $G_A \otimes G_B$.
\end{definition}

\begin{lemma}\cite{MS77}\label{lem:proddist}
The product code of codes with minimum distances $d_1$ and $d_2$ has minimum distance $d_1 d_2$.
\end{lemma}


\section[General q-ary SLRC Construction]{General $q$-ary SLRC Construction}

We will now present and analyze our construction for sequential locally recoverable codes from the code product. We will first state it in all generality before constructing explicit codes with this machinery in the following section. 



Our code construction is the $\ell$-fold product code of locally recoverable codes with the same locality parameter $r$ and is a generalization of a similar construction in~\cite{BPK16-1}, where the authors restrict to the case of $q=\ell = 2$.
The lengths, dimensions and erasure correction capabilities of the small codes can be variable.
Depending on the small codes' parameters we can determine the corresponding parameters of the product code, as we show in  Theorem \ref{thm:prod} and Corollary \ref{cor:prod}. To prove the main result we need the following lemma:

\begin{lemma}\label{lem:recvecs}
    For $i \in \{1, 2\}$, let ${\CurC}_i$ be an $(n_i, k_i, r, t_i)-$SLRC code over a field $\F_q$, and let $e_m^{(i)}$ be the $m$th unit vector in $\F_q^{n_i}$.
    \begin{enumerate}[label={(\alph*)}]
        \item For any recovery vector $c^{(1)}\in \CurC_1^\perp$ for coordinate $j\in \{1,\dots,n_1\}$ in the code $\CurC_1$ of weight $\bar r \leq r + 1$, the vector $c^{(1)} \otimes e_m^{(2)}$ is a recovery vector for coordinate $jm$ in the code $\CurC_1 \otimes \CurC_2$ of weight $\bar r$.
        \item For any recovery vector $c^{(2)}\in \CurC_2^\perp$ for coordinate $j\in \{1,\dots,n_2\}$ in the code $\CurC_2$ of weight  $\bar r\leq r + 1$, the vector $e_m^{(1)} \otimes c^{(2)}$ is a recovery vector for coordinate $(m - 1) n_2 + j$ in the code $\CurC_1 \otimes \CurC_2$ of weight $\bar r$.
    \end{enumerate}
\end{lemma}

\begin{proof}
    We first note that it is easy to see that $\wt(c^{(1)} \otimes e_m^{(2)}) = \wt(c^{(1)})$ and $  \wt(e_m^{(1)} \otimes c^{(2)})= \wt(c^{(2)})$. 
    
    Let $G_i$ be a generator matrix of $\CurC_i$, respectively. Since $c^{(1)}\in \CurC_1^\perp$, we have
    \begin{align*}
		(G_1 \otimes G_2){(c^{(1)} \otimes e_m^{(2)})}^\top
		&= (G_1 \otimes G_2)({c^{(1)}}^\top \otimes {e^{(2)}}^\top_m) \\
		&= (G_1 {c^{(1)}}^\top) \otimes (G_2 {e^{(2)}}^\top_m) \\
		&= {(0, \ldots, 0)}^\top \otimes (G_2 {e^{(2)}}^\top_m) \\
		&= {(0, \ldots, 0)}^\top,
	\end{align*}
    i.e., $c^{(1)} \otimes e_m^{(2)} \in (\CurC_1 \otimes \CurC_2)^\perp$.
    Similarly, we get that $e_m^{(1)} \otimes c^{(2)} \in (\CurC_1 \otimes \CurC_2)^\perp$, for $c^{(2)} \in \CurC_2^\perp$.

    Lastly, by application of Lemma~\ref{prop:tenvec}, if $c^{(1)}$ is a recovery vector for coordinate $j$, then $j \in \supp(c^{(1)})$.
    This implies that $(j - 1) n_2 + m \in \supp(c^{(1)} \otimes e_m^{(2)})$ and thus that $c^{(1)} \otimes e_m^{(2)}$ is a recovery vector for coordinate $(j - 1) n_2 + m$, for $m \in \{1, \ldots, n_2\}$.
    Analogously, we get that $e_m^{(1)} \otimes c^{(2)}$ is a recovery vector for coordinate $(m - 1) n_2 + j$, for $m \in \{1, \ldots, n_1\}$.
\end{proof}

\begin{theorem}\label{thm:prod}
	For all $i \in \{1, 2\}$, let ${\CurC}_i$ be an $(n_i, k_i, r, t_i)-$SLRC code over a field $\F_q$ with repair alternativity $a_i$ and let
	\begin{align*}
		N = n_1 n_2,
		&&
		K = k_1 k_2,
		&&
		T = t_1 t_2 + t_1 + t_2,
		&&
		A = a_1 + a_2.
	\end{align*}
	Then $\CurC_1 \Cpro \CurC_2$ is a $q$-ary $(N, K, r, T)$-SLRC with repair alternativity $\geq A$.
\end{theorem}

\begin{proof}
	Consider the product code of two codes $\CurC_1$ and $\CurC_2$ with generator matrices $G_1$ and $G_2$. 
	The code length $N$ and dimension $K$ of the product code follow from its definition.
    Denote the minimum distances of $\CurC_i$ by $d_i$, for $i \in \{1, 2\}$. Then we must have $d_i \geq t_i + 1$. 
    The maximum number of recoverable erasures follows from Lemma~\ref{lem:proddist}, since the minimum distance of $\CurC_1\Cpro\CurC_2$ is at least
	\begin{align*}
		d_1 d_2 - 1 \geq (t_1 + 1) (t_2 + 1) - 1=t_1t_2+t_1+t_2.
	\end{align*}  
    We now turn to the sequential local recoverability. Let $E \subseteq [n_1 n_2]$ be an erasure pattern with $|E| \leq T = t_1 t_2 + t_1 + t_2$.
    Then at least one block of coordinates $\{(i - 1) n_2 + 1, \ldots, (i - 1) n_2 + n_2\}$, for some $i\in \{1, \ldots, n_1\}$, contains at most $t_2$ erasures, since otherwise
    \begin{align*}
        |E|
        \geq n_1 (t_2 + 1)
        \geq (t_1 + k_1) (t_2 + 1)
        = t_1 t_2 + t_1 + k_1 (t_2 + 1)
        >T.
    \end{align*}
    Since this block by itself corresponds to a copy of $\CurC_2$ (which is sequentially recoverable with locality $r$) and the number of erasures inside this block is at most $t_2$, there is a coordinate $x \in \{(i - 1) n_2 + 1, \ldots, (i - 1) n_2 + n_2\}$ for which a recovery vector $v^{(2)} \in \CurC_2^\perp$ of weight at most $r+1$ exists, by Proposition \ref{prop:SLRCequiv}.
    By Lemma \ref{lem:recvecs} we know that then $ e_i^{(1)} \otimes v^{(2)}$ is a recovery vector for coordinate $x$ in $\CurC_1 \otimes \CurC_2$ of weight at most $r + 1$. 
    Therefore, again using Proposition~\ref{prop:SLRCequiv}, we have shown that the sequential locality of $\CurC_1 \otimes \CurC_2$ is (at most) $r$. 

    
%
    It remains to show the lower bound on the alternativity of the product code. 
	For $m, m' \in [n_2]$ and $c^{(1)}, c^{'(1)} \in \CurC_1^\perp$, Lemma~\ref{prop:tenvec} tells us that $\supp(c^{(1)} \otimes e_m^{(2)}) = \supp(c^{'(1)} \otimes e_{m'}^{(2)})$ if and only if $m = m'$ and $\supp(c^{(1)}) = \supp(c^{'(1)})$, so vectors in ${(\CurC_1 \Cpro \CurC_2)}^\perp$ constructed in this manner are unique up to scalar multiples.
	Similar arguments hold for vectors of the form $e_j^{(1)} \otimes c^{(2)}$ where $j \in [n_1]$ and $c^{(2)} \in \CurC_2^\perp$.
    Furthermore, 
	\begin{align*}
		&\supp(c^{(1)} \otimes e_m^{(2)})
		\cap
		\supp(e_j^{(1)} \otimes \bar c^{(2)}) \\
        =&
		\{(j' - 1) n_2 + m \mid j' \in \supp(c^{(1)})\} 
        \cap
		\{(j - 1) n_2 + m' \mid m' \in \supp(c^{(2)})\} \\
		\subseteq&
		\{(j' - 1) n_2 + m \mid j' \in [n_1]\}  \cap
		\{(j - 1) n_2 + m' \mid k' \in [n_2]\} \\
		=&
        \{(j - 1) n_2 + m\}, 
	\end{align*}
    i.e., they intersect in at most one coordinate.
	Since both $\wt(c^{(1)}) \geq 2$ and $\wt(c^{(2)}) \geq 2$, we have $\wt(c^{(1)} \otimes e_m^{(2)}) \geq 2$ and $\wt(e_j^{(1)} \otimes c^{(2)}) \geq 2$, which shows that these recovery vectors of the product code are distinct.
	For each $i \in [n_1 n_2]$, we can uniquely write $i = (j - 1) n_2 + m$ where $j \in [n_1]$ and $m \in [n_2]$.
    We then get at least $a_1$ distinct $r$-recovery sets from $\CurC_1$ corresponding to the recovery vectors $c^{'(1)} \otimes e_k^{(2)}$, where we have at least $a_1$ choices for vectors $c^{'(1)} \in \CurC_1^\perp$ for recovering the $j^\text{th}$ coordinate in $[n_1]$.
    Similarly we have $a_2$ distinct $r$-recovery sets corresponding to the recovery vectors $e_j^{(1)} \otimes c^{(2)}$ coming from $\CurC_2$. 
	Overall, the number of $r$-recovery sets of each symbol in the product is at least the sum of the number of $r$-recovery sets from each $\CurC_i$.
\end{proof}

We note that the above proof applies only to the case of sequential local recovery, not parallel local recovery. This is because the argument concerning the erasure distribution ensures only that at least one block corresponding to $\CurC_2$ contains at most $t_2$ errors. In contrast, parallel recovery requires that all such blocks contain at most $t_2$ errors. An analogous issue arises for (spread out) blocks associated with $\CurC_1$, which would each need to contain no more than $t_1$ errors.

The following corollary follows by applying Theorem~\ref{thm:prod} iteratively.

\begin{corollary}\label{cor:prod}
	For all $i \in [\ell]$, let ${\CurC}_i$ be an $(n_i, k_i, r, t_i)-$SLRC code over a field $\F_q$ with repair alternativity $a_i$ and let
	\begin{align*}
		N_\ell = \prod_{i = 1}^\ell n_i,
		&&
		K_\ell = \prod_{i = 1}^\ell k_i,
		&&
		T_\ell = \left( \prod_{i = 1}^\ell (t_i + 1) \right) - 1,
		&&
		A_\ell = \sum_{i = 1}^{\ell} a_i.
	\end{align*}
	Then ${\CurC}_1 \Cpro \cdots \Cpro {\CurC}_{\ell}$ is a $q$-ary $(N_{\ell}, K_{\ell}, r, T_{\ell})$-SLRC with repair alternativity $A \geq A_\ell$ and rate $\frac{K_\ell}{N_\ell}$.
\end{corollary}

We can consider the product of $\ell$ codes in a geometric sense by viewing the resulting product code as an $n_1 \times \cdots \times n_\ell$ hyperrectangle whose grid points represent the code's $n_1 \cdots n_\ell$ symbols.
Then, given some index $i \in [n_1 \cdots n_\ell]$, the $\ell$ orthogonal lines within this hyperrectangle containing that corresponding point correspond to the codes $\CurC_1, \ldots \CurC_\ell$ and any other $r$ symbols on each line can then be used as a recovery set for $i$.
For example, Figure~\ref{figure:prod} shows a product between an $n_1$ length code (represented by the horizontal lines) and an $n_2$ length code (represented by the vertical lines).
Any erased vertex can then be recovered from a subset of vertices sharing either a vertical or horizontal line with it.

\begin{figure}[ht]
	\begin{center}
		\begin{tikzpicture}[scale=0.75]
			\GraphInit[vstyle=Classic]
			\SetGraphUnit{1}
			\SetUpVertex[Math, MinSize=1pt]

			\def\s{1.0} 

			\Vertex[x={\s*0}, y={\s*5}, L={1}, Lpos=180]{a04e}
			\Vertex[x={\s*0}, y={\s*4}, L={2}, Lpos=180]{a03}
			\Vertex[x={\s*0}, y={\s*3}, L={3}, Lpos=180]{a02}
			\Vertex[x={\s*0}, y={\s*1}, L={n_2 - 1}, Lpos=180]{a01}
			\Vertex[x={\s*0}, y={\s*0}, L={n_2}, Lpos=180]{a00}

			\Vertex[x={\s*0}, y={\s*5}, L={1}, Lpos=90]{a04}
			\Vertex[x={\s*1}, y={\s*5}, L={2}, Lpos=90]{a14}
			\Vertex[x={\s*2}, y={\s*5}, L={3}, Lpos=90]{a24}
			\Vertex[x={\s*4}, y={\s*5}, L={n_1 - 1}, Lpos=90]{a34}
			\Vertex[x={\s*5}, y={\s*5}, L={n_1}, Lpos=90]{a44}

			\Vertex[x={\s*0}, y={\s*4}, NoLabel]{a03}
			\Vertex[x={\s*1}, y={\s*4}, NoLabel]{a13}
			\Vertex[x={\s*2}, y={\s*4}, NoLabel]{a23}
			\Vertex[x={\s*4}, y={\s*4}, NoLabel]{a33}
			\Vertex[x={\s*5}, y={\s*4}, NoLabel]{a43}

			\Vertex[x={\s*0}, y={\s*3}, NoLabel]{a02}
			\Vertex[x={\s*1}, y={\s*3}, NoLabel]{a12}
			\Vertex[x={\s*2}, y={\s*3}, NoLabel]{a22}
			\Vertex[x={\s*4}, y={\s*3}, NoLabel]{a32}
			\Vertex[x={\s*5}, y={\s*3}, NoLabel]{a42}

			\Vertex[x={\s*0}, y={\s*1}, NoLabel]{a01}
			\Vertex[x={\s*1}, y={\s*1}, NoLabel]{a11}
			\Vertex[x={\s*2}, y={\s*1}, NoLabel]{a21}
			\Vertex[x={\s*4}, y={\s*1}, NoLabel]{a31}
			\Vertex[x={\s*5}, y={\s*1}, NoLabel]{a41}

			\Vertex[x={\s*0}, y={\s*0}, NoLabel]{a00}
			\Vertex[x={\s*1}, y={\s*0}, NoLabel]{a10}
			\Vertex[x={\s*2}, y={\s*0}, NoLabel]{a20}
			\Vertex[x={\s*4}, y={\s*0}, NoLabel]{a30}
			\Vertex[x={\s*5}, y={\s*0}, NoLabel]{a40}

			\tikzstyle{EdgeStyle}=[solid, color=black]
			\Edges(a00, a10, a20)
			\Edges(a01, a11, a21)
			\Edges(a02, a12, a22)
			\Edges(a03, a13, a23)
			\Edges(a04, a14, a24)

			\Edges(a30, a40)
			\Edges(a31, a41)
			\Edges(a32, a42)
			\Edges(a33, a43)
			\Edges(a34, a44)

			\tikzstyle{EdgeStyle}=[solid, color=red]
			\Edges(a02, a03, a04)
			\Edges(a12, a13, a14)
			\Edges(a22, a23, a24)
			\Edges(a32, a33, a34)
			\Edges(a42, a43, a44)

			\Edges(a00, a01)
			\Edges(a10, a11)
			\Edges(a20, a21)
			\Edges(a30, a31)
			\Edges(a40, a41)

			\tikzstyle{EdgeStyle}=[dashed, color=black]
			\Edges(a20, a30)
			\Edges(a21, a31)
			\Edges(a22, a32)
			\Edges(a23, a33)
			\Edges(a24, a34)

			\tikzstyle{EdgeStyle}=[dashed, color=red]
			\Edges(a01, a02)
			\Edges(a11, a12)
			\Edges(a21, a22)
			\Edges(a31, a32)
			\Edges(a41, a42)
		\end{tikzpicture}
	\end{center}
	\caption{\centering{The product of an $n_1$ length code with an $n_2$ length code.}}\label{figure:prod}
\end{figure}

For a geometric proof of Corollary~\ref{cor:prod}, see Appendix~\ref{app:geoprod}.

\begin{remark}
The construction of recovery vectors for the product code from those of the component codes naturally leads to an algorithmic procedure for erasure recovery. Given an erasure pattern, one identifies projections onto component codes---corresponding to lines in the ambient hyperrectangle---such that the erased coordinates along these lines can be recovered locally via the associated recovery vectors. This process is applied iteratively, updating the erasure pattern after each successful recovery step, until all erasures are resolved.
\end{remark}





By choosing codes of the same dimension, we can show that the alternativity can be explicitly computed (instead of just being subject to a lower bound).
First we need the following lemma.

\begin{lemma}\label{lemma:linind2}
	Let $r$, $k$, $k'$, $n$, and $n'$ be positive integers such that $1 < r \leq k < n$ and $1 < r \leq k' < n'$.
	Let $C$ be a $k \times n$ matrix and let $D$ be an $k' \times n'$ matrix where both matrices have the property that any $r$ columns are linearly independent.
	Let $\{(c_1, d_1), \ldots, (c_{r + 1}, d_{r + 1})\}$ be a set of pairs of columns of $C$ and $D$ such that for all $i, i' \in [r + 1]$ where $i \neq i'$, both $c_i \neq c_{i'}$ and $d_i \neq d_{i'}$.
	Then the set of vectors $\{c_1 \otimes d_1, \ldots, c_{r + 1} \otimes d_{r + 1}\}$ are linearly independent.
\end{lemma}

\begin{proof} 
    If all $r + 1$ vectors in either the set $\{c_1, \ldots, c_{r + 1}\}$ or the set $\{d_1, \ldots, d_{r + 1}\}$ are already linearly independent, then so are the vectors in $\{c_1 \otimes d_1, \ldots, c_{r + 1} \otimes d_{r + 1}\}$ and we are done.
    
    So, assume that $c_{r + 1}$ and $d_{r + 1}$ are respectively linear combinations of $c_1, \ldots, c_r$ and $d_1, \ldots, d_r$.
    It follows by the linear independence of any $r$ columns of $C$ or $D$ that these sums must necessarily include all $r$ other vectors, i.e.,
    \begin{align*}
        c_{r + 1} = \sum_{i = 1}^r a_i c_i, \qquad d_{r + 1} = \sum_{i = 1}^r b_i d_i,
    \end{align*}
    for $a_i, b_i \in \F_q \setminus \{0\}$ for all $i \in [r]$.
	Now, let $A$ and $B$ be invertible $k \times k$ matrices such that for all $1 \leq i \leq r$, $A c_i = e_i \in \F_q^k$, $B d_i = e_i \in \F_q^{k'}$.
    It follows that
	\begin{align*}
		A c_{r + 1}
		&= \begin{bmatrix}
			\smash{\underbrace{\begin{matrix} a_1 & \cdots & a_r \end{matrix}}_{r}}
			&
			\smash{\underbrace{\begin{matrix} 0 & \cdots & 0 \end{matrix}}_{k - r}}
		\end{bmatrix}^\top,
		\vphantom{= \begin{bmatrix}
			\underbrace{\begin{matrix} a_1 & \cdots & a_r \end{matrix}}_{r}
			&
			\underbrace{\begin{matrix} 0 & \cdots & 0 \end{matrix}}_{k - r}
		\end{bmatrix}^\top,}
		\\
		B d_{r + 1}
		&= \begin{bmatrix}
			\smash{\underbrace{\begin{matrix} b_1 & \cdots & b_r \end{matrix}}_{r}}
			&
			\smash{\underbrace{\begin{matrix} 0 & \cdots & 0 \end{matrix}}_{k' - r}}
		\end{bmatrix}^\top.
		\vphantom{= \begin{bmatrix}
			\underbrace{\begin{matrix} b_1 & \cdots & b_r \end{matrix}}_{r}
			&
			\underbrace{\begin{matrix} 0 & \cdots & 0 \end{matrix}}_{k' - r}
		\end{bmatrix}^\top.}
	\end{align*}
	Then for all $1 \leq i \leq r$,
	\begin{align*}
		(A \otimes B) (c_i \otimes d_i)
		= (A c_i) \otimes (B d_i)
		= e_{(i - 1) r + i} \in \F_q^{kk'}, \\
		(A \otimes B) (c_{r + 1} \otimes d_{r + 1})
		= \begin{bmatrix}
			\smash{\underbrace{\begin{matrix} a_1 b_1 & \cdots & a_r b_r \end{matrix}}_{r^2}}
			&
			\smash{\underbrace{\begin{matrix} 0 & \cdots & 0 \end{matrix}}_{k k' - r^2}}
		\end{bmatrix}^\top \in \F_q^{kk'},
		\vphantom{= \begin{bmatrix}
			\underbrace{\begin{matrix} 1 & \cdots & 1 \end{matrix}}_{r^2}
			&
			\underbrace{\begin{matrix} 0 & \cdots & 0 \end{matrix}}_{n n' - r^2}
		\end{bmatrix}^\top \in \F_q^{kk'},}
	\end{align*}
    where the $r^2$ values $a_1 b_1, \ldots, a_r b_r$ will be nonzero.
	Since
	\begin{align*}
		\abs{\cup_{i \in [r]} \supp((A \otimes B) (c_i \otimes d_i))} = r < r^2 = \abs{\supp((A \otimes B) (c_{r + 1} \otimes d_{r + 1}))},
	\end{align*}
	these $r + 1$ vectors $\{(A \otimes B)(c_i \otimes d_i)\}_{i \in [r + 1]}$ in $\F_q^{kk'}$ are all linearly independent and so then too are the vectors ${\{c_i \otimes d_i\}}_{i \in [r + 1]}$ because $A \otimes B$ is full rank.
\end{proof}

\begin{theorem}\label{thm:mdsalt}
	Given some fixed $r > 1$, if for all $i \in [\ell]$, $\CurC_i$ is a $[n_i, k]$ linear code over $\F_q$ with locality $r$ and repair alternativity $a_i$, then the product $\CurC_1 \Cpro \cdots \Cpro \CurC_\ell$ is an SLRC with locality $r$, with repair alternativity $\sum_{i = 1}^\ell a_i$, and whose dual code's minimum distance is at most $r + 1$.
\end{theorem}

\begin{proof}
	Proceed by induction on $\ell$.
	This is trivial for $\ell = 1$.
	Assume $\CurC = \CurC_1 \Cpro \cdots \Cpro \CurC_{\ell - 1}$ is an $(n, k^{\ell - 1}, r, t)$-SLRC with $d(\CurC^\perp) \leq r + 1$, locality $r$, repair alternativity $a$, and generator matrix $C$.
	Let $\CurD$ be a $[n', k']$ linear code with locality $r$, repair alternativity $a'$, and generator matrix $D$.

	Now consider the product code $\CurC \Cpro \CurD$ with generator matrix $C \otimes D$.
	From the proof of Theorem~\ref{thm:prod}, we know that for each $i \in [n n']$, there are at least $(a + a')(q - 1)$ parity checks in the product of the form $\lambda (c \otimes e_j)$ or $\lambda (e_i \otimes d)$ where $\lambda \in \F_q^\times$, $c \in \CurC^\perp$ with $\wt(c) \leq r + 1$, $d \in \CurD^\perp$ with $\wt(d) = r + 1$, and $e_i \in \F_q^n$ and $e_j \in \F_q^{n'}$ are standard basis vectors.
	Assume that there exists $p \in {(\CurC \Cpro \CurD)}^\perp$ with $1 \leq \wt(p) \leq r + 1$ which is not of the form above.
	Then $(C \otimes D) p^\top = 0$.
	Write
	\begin{align*}
		p = (p_{11}, \ldots, p_{1n'}, p_{2,1}, \ldots, p_{n1}, \ldots, p_{n n'})
	\end{align*}
	and let 
	\begin{align*}
		P = \{(i, j) \in [n] \times [n'] \mid p_{ij} \neq 0\}.
	\end{align*}
	Note that because $p$ is not of the form $\lambda (c \otimes e_j)$ or $\lambda (e_i \otimes d)$, neither all of the $i$s nor all of the $j$s can take the same value.

	If $\gamma \in [k k']$ is the index of a row in the generator matrix of the product code, we can write $\gamma = (\alpha - 1) k' + \beta$ for $\alpha \in [k]$ and $\beta \in [k']$ and the $\gamma^\text{th}$ row of $(C \otimes D) p^\top$ is equal to 
	\begin{align*}
		0
		= \sum_{i = 1}^n c_{\alpha i} \sum_{j = 1}^{n'} d_{\beta j} p_{ij} 
		= \sum_{j = 1}^{n'} d_{\beta j} \sum_{i = 1}^{n} c_{\alpha i} p_{ij}.
	\end{align*}
	We then have that for all $\alpha \in [k]$ and all $\beta \in [k']$
	\begin{align*}
		\left(
			\sum_{j = 1}^{n'} d_{\beta j} p_{1 j},
			\ldots,
			\sum_{j = 1}^{n'} d_{\beta j} p_{n j}
		\right)
		&\in \CurC^\perp, \\
		\left(
			\sum_{i = 1}^{n} c_{\alpha i} p_{i 1},
			\ldots,
			\sum_{i = 1}^{n} c_{\alpha i} p_{i n'}
		\right)
		&\in \CurD^\perp.
	\end{align*}
	Because we have at least two distinct values of $i$ and $j$ in $P$, each sum has less than $r + 1$ nonzero terms and thus these vectors are never all zeros.
	Furthermore, because the minimum distance of both $\CurC^\perp$ and $\CurD^\perp$ are at most $r + 1$, it follows that all values of $i$ and $j$ in $P$ must be distinct.
	In other words, for all distinct pairs $(i, j), (i', j') \in P$, both $i \neq i'$ and $j \neq j'$.
	Let ${\{(c_i, d_j)\}}_{(i, j) \in P}$ be pairs of the appropriate $i^\text{th}$ and $j^\text{th}$ columns of $C$ and $D$ and apply Lemma~\ref{lemma:linind1} and Lemma~\ref{lemma:linind2}.
	It then follows that $(C \otimes D) p^\top \neq 0$ because $p$ is nonzero and the columns of $C \otimes D$ corresponding to the nonzero entries of $p$ are all linearly independent.
	This is a contradiction.
\end{proof}


\section[Product Codes from MDS and BCH Codes]{Product Codes from MDS and BCH Codes}\label{sec:examples}

In this section, we derive several examples to demonstrate the achievable parameters of our construction by using MDS and BCH codes as the building blocks for the product code. 

For this let $\CurP_q$ be the ${[q + 1, 2, q]}_q$ MDS code, $\CurD_{q,n}$ be the ${[n, n - 1, 2]}_q$ MDS code, and $\CurR_{q,n,k}$ be the ${[n, k, n - k + 1]}_q$ Reed-Solomon code.
Furthermore, denote by $\CurB_{q, n, d}$ the BCH code over $\F_q$ with length $n$ and designed distance $d$.
Finally, for some code $\CurC$, let $\CurC^{*S}$ denote the code $\CurC$ punctured at indices $S$.

The achievable parameters are presented in the tables below, where we omit certain product codes when a superior code exists for the same values of \( q \) and \( k \). By "superior," we refer to codes that exhibit one or more of the following characteristics: larger \( \ell \), smaller length \( n \), smaller locality \( r \), or a higher number of recoverable erasures \( t \). For instance, in \( q = 5 \), we omit the \( (16,6,3,5) \)-SLRC because the \( (15,6,3,5) \)-SLRC achieves the same parameters with a smaller length and lower alternativity.

Additionally, we omit constructions from codes whose dual codes have minimum distance 2 since in these cases both Theorem~\ref{thm:mdsalt} no longer applies and, in most cases given the parameters of the resulting product codes, optimal codes are known to exist~\cite{GXY19}.

Table~\ref{table:q3} gives parameter sets for products of MDS and (possibly punctured) BCH codes over $\F_3$ up to $k = 10$ and these are additionally displayed in Figure~\ref{chart:q3}.
Table~\ref{table:q5k46}, Table~\ref{table:q5k8}, and Table~\ref{table:q5k910} give parameter sets of products of MDS and (possibly punctured) BCH codes over $\F_5$ up to $k = 10$; a chart of these resulting codes is given in Figure~\ref{chart:q5}.

\begin{table}[ht]
	\centering{}
	\begin{tabular}{rrrrlrrl}
		\toprule{}
		$q$ & $k$ & $r$ & $\ell$ & Code & $t$ & $a$ & Construction \\
		\midrule{}
		$3$ & $4$ & $2$ & $2$
		   & $[16,  4,  9]$ &  8 & 6 & $\CurP_3		\Cpro \CurP_3$	 \\
		&&&& $[12,  4,  6]$ &  5 & 4 & $\CurP_3		\Cpro \CurD_{3,3}$ \\
		&&&& $[9,   4,  4]$ &  3 & 2 & $\CurD_{3,3}	\Cpro \CurD_{3,3}$ \\
		\midrule{}
		$3$ & $6$ & $3$ & $2$
		   & $[16,  6,  6]$ &  5 & 4 & $\CurP_3		\Cpro \CurD_{3,4}$ \\
		&&&& $[12,  6,  4]$ &  3 & 2 & $\CurD_{3,4}	\Cpro \CurD_{3,3}$ \\
		\midrule{}
		$3$ & $8$ & $4$ & $2$
		   & $[20,  8,  6]$ &  5 & 4 & $\CurP_3	   \Cpro \CurD_{3,5}$ \\
		&&&& $[15,  8,  4]$ &  3 & 2 & $\CurD_{3,5}   \Cpro \CurD_{3,3}$ \\
		\midrule{}
		$3$ & $8$ & $3$ & $2$
		   & $[32,  8, 12]$ & 11 & 8 & $\CurP_3	   \Cpro \CurB_{3,8,4}$ \\
		&&&& $[28,  8,  9]$ &  8 & 5 & $\CurP_3	   \Cpro \CurB_{3,8,4}^{*\{1\}}$ \\
		&&&& $[24,  8,  8]$ &  7 & 6 & $\CurD_{3,3}   \Cpro \CurB_{3,8,4}$ \\
		&&&& $[21,  8,  6]$ &  5 & 3 & $\CurD_{3,3}   \Cpro \CurB_{3,8,4}^{*\{1\}}$ \\
		&&&& $[18,  8,  4]$ &  3 & 3 & $\CurD_{3,3}   \Cpro \CurB_{3,8,4}^{*\{1,5\}}$ \\
		\midrule{}
		$3$ & $8$ & $2$ & $3$
		   & $[64,  8, 27]$ & 26 & 9 & $\CurP_3	   \Cpro \CurP_3	 \Cpro \CurP_3$	 \\
		&&&& $[48,  8, 18]$ & 17 & 7 & $\CurP_3	   \Cpro \CurP_3	 \Cpro \CurD_{3,3}$ \\
		&&&& $[36,  8, 12]$ & 11 & 5 & $\CurP_3	   \Cpro \CurD_{3,3} \Cpro \CurD_{3,3}$ \\
		&&&& $[27,  8,  8]$ &  7 & 3 & $\CurD_{3,3}   \Cpro \CurD_{3,3} \Cpro \CurD_{3,3}$ \\
		\midrule{}
		$3$ & $9$ & $3$ & $2$
		   & $[16,  9,  4]$ &  3 & 2 & $\CurD_{3,4}   \Cpro \CurD_{3,4}$ \\
		\midrule{}
		$3$ & $10$ & $5$ & $2$
		   & $[24, 10,  6]$ &  5 & 4 & $\CurP_3	   \Cpro \CurD_{3,6}$ \\
		&&&& $[18, 10,  4]$ &  3 & 2 & $\CurD_{3,6}   \Cpro \CurD_{3,3}$ \\
		\bottomrule{}
	\end{tabular}
	\caption{\centering{Parameters of constructions using $q = 3, k \leq 10$}}\label{table:q3}
\end{table}

\begin{figure}[ht]
	\centering{}
	\includegraphics[width=6cm]{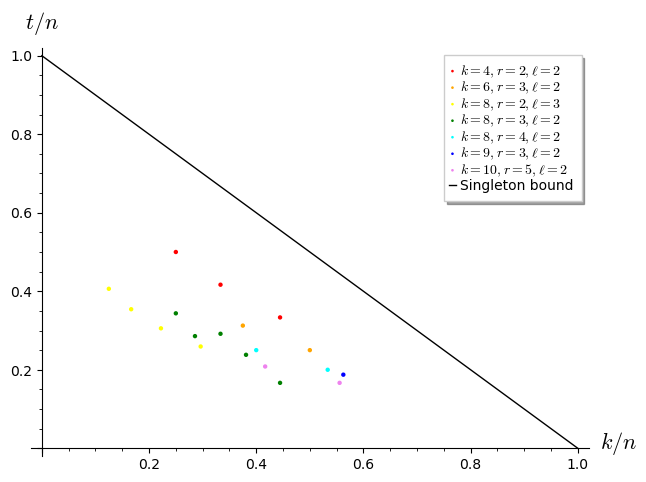}
	\caption{Product codes for $q = 3, k \leq 10$}\label{chart:q3}
\end{figure}

\begin{table}[ht]
	\centering{}
	\begin{tabular}{rrrrlrrl}
		\toprule{}
		$q$ & $k$ & $r$ & $\ell$ & Code & $t$ & $a$ & Construction \\
		\midrule{}
		$5$ & $4$ & $2$ & $2$
		   & $[36, 4, 25]$ & 24 & 20 & $\CurP_5	   \Cpro \CurP_5$	   \\
		&&&& $[30, 4, 20]$ & 19 & 16 & $\CurP_5	   \Cpro \CurR_{5,5,2}$ \\
		&&&& $[25, 4, 16]$ & 15 & 12 & $\CurR_{5,5,2} \Cpro \CurR_{5,5,2}$ \\
		&&&& $[24, 4, 15]$ & 14 & 13 & $\CurP_5	   \Cpro \CurR_{5,4,2}$ \\
		&&&& $[20, 4, 12]$ & 11 &  9 & $\CurR_{5,5,2} \Cpro \CurR_{5,4,2}$ \\
		&&&& $[18, 4, 10]$ &  9 & 11 & $\CurP_5	   \Cpro \CurD_{5,3}$   \\
		&&&& $[16, 4,  9]$ &  8 &  6 & $\CurR_{5,4,2} \Cpro \CurR_{5,4,2}$ \\
		&&&& $[15, 4,  8]$ &  7 &  7 & $\CurR_{5,5,2} \Cpro \CurD_{5,3}$   \\
		&&&& $[12, 4,  6]$ &  5 &  4 & $\CurR_{5,4,2} \Cpro \CurD_{5,3}$   \\
		&&&& $[ 9, 4,  4]$ &  3 &  2 & $\CurD_{5,3}   \Cpro \CurD_{5,3}$   \\
		\midrule{}
		$5$ & $6$ & $3$ & $2$
		   & $[30, 6, 15]$ & 14 & 14 & $\CurP_5	   \Cpro \CurR_{5,5,3}$ \\
		&&&& $[25, 6, 12]$ & 11 & 10 & $\CurR_{5,5,3} \Cpro \CurR_{5,5,2}$ \\
		&&&& $[24, 6, 10]$ &  9 & 11 & $\CurP_5	   \Cpro \CurD_{5,4}$   \\
		&&&& $[20, 6,  9]$ &  8 &  7 & $\CurR_{5,5,3} \Cpro \CurR_{5,4,2}$ \\
		&&&& $[15, 6,  6]$ &  5 &  5 & $\CurR_{5,5,3} \Cpro \CurD_{5,3}$   \\
		&&&& $[12, 6,  4]$ &  3 &  2 & $\CurD_{5,4}   \Cpro \CurD_{5,3}$   \\
		\bottomrule{}
	\end{tabular}
	\caption{\centering{Parameters of constructions using $q = 5, k \in \{4, 6\}$}}\label{table:q5k46}
\end{table}

\begin{table}[ht]
	\centering{}
	\begin{tabular}{rrrrlrrl}
		\toprule{}
		$q$ & $k$ & $r$ & $\ell$ & Code & $t$ & $a$ & Construction \\
		\midrule{}
		$5$ & $8$ & $4$ & $2$
		   & $[ 30, 8, 10]$ &  9 & 11 & $\CurP_5	   \Cpro \CurD_{5,5}$ \\
		&&&& $[ 25, 8,  8]$ &  7 &  7 & $\CurR_{5,5,2} \Cpro \CurD_{5,5}$ \\
		&&&& $[ 20, 8,  6]$ &  5 &  4 & $\CurR_{5,4,2} \Cpro \CurD_{5,5}$ \\
		&&&& $[ 15, 8,  4]$ &  3 &  2 & $\CurD_{5,5}   \Cpro \CurD_{5,3}$ \\
		\midrule{}
		$5$ & $8$ & $2$ & $3$
		   & $[100, 8, 48]$ & 47 & 15 & $\CurR_{5,5,2} \Cpro \CurR_{5,5,2} \Cpro \CurR_{5,4,2}$ \\
		&&&& $[ 96, 8, 45]$ & 44 & 19 & $\CurP_5	   \Cpro \CurR_{5,5,2} \Cpro \CurR_{5,4,2}$ \\
		&&&& $[ 90, 8, 40]$ & 39 & 17 & $\CurP_5	   \Cpro \CurR_{5,5,2} \Cpro \CurD_{5,3}$   \\
		&&&& $[ 80, 8, 36]$ & 35 & 12 & $\CurR_{5,5,2} \Cpro \CurR_{5,4,2} \Cpro \CurR_{5,4,2}$ \\
		&&&& $[ 75, 8, 32]$ & 31 & 13 & $\CurR_{5,5,2} \Cpro \CurR_{5,5,2} \Cpro \CurD_{5,3}$   \\
		&&&& $[ 72, 8, 30]$ & 29 & 14 & $\CurP_5	   \Cpro \CurR_{5,4,2} \Cpro \CurD_{5,3}$   \\
		&&&& $[ 64, 8, 27]$ & 26 &  9 & $\CurR_{5,4,2} \Cpro \CurR_{5,4,2} \Cpro \CurR_{5,4,2}$ \\
		&&&& $[ 60, 8, 24]$ & 23 & 10 & $\CurR_{5,5,2} \Cpro \CurR_{5,4,2} \Cpro \CurD_{5,3}$   \\
		&&&& $[ 54, 8, 20]$ & 19 & 12 & $\CurP_5	   \Cpro \CurD_{5,3}   \Cpro \CurD_{5,3}$   \\
		&&&& $[ 48, 8, 18]$ & 17 &  7 & $\CurR_{5,4,2} \Cpro \CurR_{5,4,2} \Cpro \CurD_{5,3}$   \\
		&&&& $[ 45, 8, 16]$ & 15 &  8 & $\CurR_{5,5,2} \Cpro \CurD_{5,3}   \Cpro \CurD_{5,3}$   \\
		&&&& $[ 36, 8, 12]$ & 11 &  5 & $\CurR_{5,4,2} \Cpro \CurD_{5,3}   \Cpro \CurD_{5,3}$   \\
		&&&& $[ 27, 8,  8]$ &  7 &  3 & $\CurD_{5,3}   \Cpro \CurD_{5,3}   \Cpro \CurD_{5,3}$   \\
		\bottomrule{}
	\end{tabular}
	\caption{\centering{Parameters of constructions using $q = 5, k = 8$}}\label{table:q5k8}
\end{table}

\begin{table}[ht]
	\centering{}
	\begin{tabular}{rrrrlrrl}
		\toprule{}
		$q$ & $k$ & $r$ & $\ell$ & Code & $t$ & $a$ & Construction \\
		\midrule{}
		$5$ & $9$ & $3$ & $2$
		   & $[25, 9,  9]$ & 8 &  8 & $\CurR_{5,5,3} \Cpro \CurR_{5,5,3}$ \\
		&&&& $[20, 9,  6]$ & 5 &  5 & $\CurR_{5,5,3} \Cpro \CurD_{5,4}$   \\
		&&&& $[16, 9,  4]$ & 3 &  2 & $\CurD_{5,4}   \Cpro \CurD_{5,4}$   \\
		\midrule{}
		$5$ & $10$ & $5$ & $2$
		   & $[42, 10, 10]$ & 9 & 13 & $\CurP_5	   \Cpro \CurB_{5,8,3}^{*\{1\}}$ \\
		&&&& $[36, 10, 10]$ & 9 & 11 & $\CurP_5	   \Cpro \CurD_{5,6}$ \\
		&&&& $[35, 10,  8]$ & 7 &  9 & $\CurR_{5,5,2} \Cpro \CurB_{5,8,3}^{*\{1\}}$ \\
		&&&& $[30, 10,  8]$ & 7 &  7 & $\CurR_{5,5,2} \Cpro \CurD_{5,6}$ \\
		&&&& $[28, 10,  6]$ & 5 &  6 & $\CurR_{5,4,2} \Cpro \CurB_{5,8,3}^{*\{1\}}$ \\
		&&&& $[24, 10,  6]$ & 5 &  4 & $\CurR_{5,4,2} \Cpro \CurD_{5,6}$ \\
		&&&& $[21, 10,  4]$ & 3 &  4 & $\CurD_{5,3}   \Cpro \CurB_{5,8,3}^{*\{1\}}$ \\
		&&&& $[18, 10,  4]$ & 3 &  2 & $\CurD_{5,6}   \Cpro \CurD_{5,3}$ \\
		\bottomrule{}
	\end{tabular}
	\caption{\centering{Parameters of constructions using $q = 5, k \in \{9, 10\}$}}\label{table:q5k910}
\end{table}

\begin{figure}[ht]
	\centering{}
	\includegraphics[width=6cm]{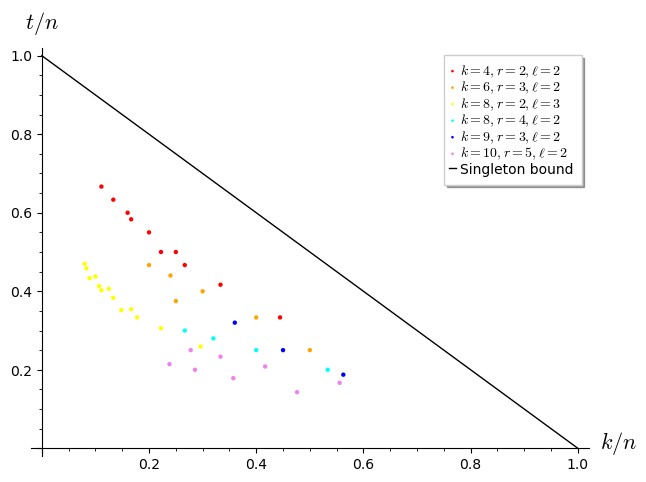}
	\caption{\centering{Product codes for $q = 5, k \leq 10$}}\label{chart:q5}
\end{figure}


\subsubsection*{Explicit examples}

\begin{example}
	Let $\CurP_3$ be the ${[4, 2, 3]}_3$ MDS code and let $\CurD_{3,3}$ be the ${[3, 2, 2]}_3$ MDS code and let $P$ and $D$ denote their respective generator matrices
	\begin{align*}
		P = \begin{bmatrix}
			1 & 0 & 1 & 2 \\
			0 & 1 & 1 & 1
		\end{bmatrix},
		&&
		D = \begin{bmatrix}
			1 & 0 & 1 \\
			0 & 1 & 1
		\end{bmatrix}.
	\end{align*}
	Since both are MDS with dimension $k = 2$, they both have the same locality $r = 2$, $\CurP_3$ has alternativity $\binom{3}{2} = 3$, and $\CurD_{3,3}$ has alternativity $\binom{2}{2} = 1$.
	Then the product code is a $(12, 4, 2, 5)$-SLRC with alternativity $4$ and generator matrix
	\begin{align*}
		P \otimes D = \begin{bmatrix}
			1 & 0 & 1 & 0 & 0 & 0 & 1 & 0 & 1 & 2 & 0 & 2 \\
			0 & 1 & 1 & 0 & 0 & 0 & 0 & 1 & 1 & 0 & 2 & 2 \\
			0 & 0 & 0 & 1 & 0 & 1 & 1 & 0 & 1 & 1 & 0 & 1 \\
			0 & 0 & 0 & 0 & 1 & 1 & 0 & 1 & 1 & 0 & 1 & 1
		\end{bmatrix}.
	\end{align*}
	We can represent this $q$-ary SLRC as in Figure~\ref{figure:(12,4,2,5)-SLRC]} where each node represents a code symbol.
	When a node is erased, it is able to be either recovered from the two nodes on the dotted line containing that, or from any pair of nodes with which it shares a solid line.
	In other words, the dotted lines represent the code $\CurD_{3,3}$, and the solid lines represent the code $\CurP_3$.

	\begin{figure}[ht]
		\begin{center}
			\begin{tikzpicture}
				\GraphInit[vstyle=Classic]
				\SetGraphUnit{1}
				\SetUpVertex[Math, MinSize=1pt]
	
				\def\s{1.0} 
	
				\foreach \x in {0,...,3}
				\foreach \y in {0,...,2}{%
					\Vertex[%
						x={\s*\x},
						y={\s*\y},
						NoLabel
					]{a\x\y}
				}
	
				\tikzstyle{EdgeStyle}=[solid, color=black]
				\Edges(a00, a01, a02)
				\Edges(a10, a11, a12)
				\Edges(a20, a21, a22)
				\Edges(a30, a31, a32)

				\tikzstyle{EdgeStyle}=[solid, color=green]
				\Edges[style={bend left }](a10, a20)
				\Edges[style={bend right}](a20, a30)
				\Edges[style={bend left }](a11, a21)
				\Edges[style={bend right}](a21, a31)
				\Edges[style={bend left }](a12, a22)
				\Edges[style={bend right}](a22, a32)

				\tikzstyle{EdgeStyle}=[solid, color=red]
				\Edges[style={bend right=45}](a00, a20)
				\Edges[style={bend left	}](a20, a30)
				\Edges[style={bend right=45}](a01, a21)
				\Edges[style={bend left	}](a21, a31)
				\Edges[style={bend right=45}](a02, a22)
				\Edges[style={bend left	}](a22, a32)

				\tikzstyle{EdgeStyle}=[solid, color=blue]
				\Edges[style={bend right=45}](a00, a10)
				\Edges[style={bend left	}](a10, a30)
				\Edges[style={bend right=45}](a01, a11)
				\Edges[style={bend left	}](a11, a31)
				\Edges[style={bend right=45}](a02, a12)
				\Edges[style={bend left	}](a12, a32)

				\tikzstyle{EdgeStyle}=[solid, color=orange]
				\Edges[style={bend left }](a00, a10)
				\Edges[style={bend right}](a10, a20)
				\Edges[style={bend left }](a01, a11)
				\Edges[style={bend right}](a11, a21)
				\Edges[style={bend left }](a02, a12)
				\Edges[style={bend right}](a12, a22)
	
			\end{tikzpicture}
		\end{center}
		\caption{\centering{A $(12,4,2,5)$-SLRC over $\F_3$.}}\label{figure:(12,4,2,5)-SLRC]}
	\end{figure}
\end{example}

\begin{example}
	Again, using the ${[4, 2, 3]}_3$ code $\CurP_3$ above, consider the product code $\CurP_3 \Cpro \CurP_3 \Cpro \CurP_3$.
	This is a $q$-ary $(64, 8, 2, 26)$-SLRC with alternativity $9$.
	Similar to before, we can represent this SLRC graphically: see Figure~\ref{figure:(64,8,2,26)-SLRC}.
	Each color (red, green, and blue) represents one of the three copies of $\CurP_3$ in the product where any node can be recovered using any two other nodes from the same line containing the erased node.

	\begin{figure}[ht]
		\begin{center}
			\begin{tikzpicture}[scale=0.75]
				\GraphInit[vstyle=Classic]
				\SetGraphUnit{1}
				\SetUpVertex[Math, MinSize=1pt]

				\def\s{2}	
				\def\zr{0.3} 
				\def\za{30}  

				\foreach \x in {0,...,3}
				\foreach \y in {0,...,3}
				\foreach \z in {0,...,3}{%
					\Vertex[%
						x={\s*\x + cos{\za}*\s*\zr*\z},
						y={\s*\y + sin{\za}*\s*\zr*\z},
						NoLabel
					]{a\x\y\z}
				}

				\tikzstyle{EdgeStyle}=[solid, color=red]
				\foreach \y in {0,...,3}{%
					\Edges(a0\y0, a1\y0, a2\y0, a3\y0)
				}
				\foreach \z in {0,...,3}{%
					\Edges(a03\z, a13\z, a23\z, a33\z)
				}

				\tikzstyle{EdgeStyle}=[dotted, color=red]
				\foreach \y in {0,...,2}
				\foreach \z in {1,...,3}{%
					\Edges(a0\y\z, a1\y\z, a2\y\z, a3\y\z)
				}

				\tikzstyle{EdgeStyle}=[solid, color=blue]
				\foreach \x in {0,...,3}{%
					\Edges(a\x00, a\x10, a\x20, a\x30)
				}
				\foreach \z in {0,...,3}{%
					\Edges(a30\z, a31\z, a32\z, a33\z)
				}

				\tikzstyle{EdgeStyle}=[dotted, color=blue]
				\foreach \x in {0,...,2}
				\foreach \z in {1,...,3}{%
					\Edges(a\x0\z, a\x1\z, a\x2\z, a\x3\z)
				}

				\tikzstyle{EdgeStyle}=[solid, color=green]
				\foreach \x in {0,...,3}{%
					\Edges(a\x30, a\x31, a\x32, a\x33)
				}
				\foreach \y in {0,...,3}{%
					\Edges(a3\y0, a3\y1, a3\y2, a3\y3)
				}

				\tikzstyle{EdgeStyle}=[dotted, color=green]
				\foreach \x in {0,...,2}
				\foreach \y in {0,...,2}{%
					\Edges(a\x\y0, a\x\y1, a\x\y2, a\x\y3)
				}
			\end{tikzpicture}
		\end{center}
		\caption{\centering{A ternary $(64, 8, 2, 26)$-SLRC.}}\label{figure:(64,8,2,26)-SLRC}
	\end{figure}
\end{example}

\begin{example}\label{ex:recovery}
	Consider the product $\CurC = \CurR_{5,5,2} \Cpro \CurR_{5,5,2}$ which is a $(25, 4, 2, 15)$-SLRC.\@
	In Figure~\ref{figure:iterations}, black vertices represent un-erased symbols, white represent erased symbols.
	In order to fully recover the entire codeword, recovery must proceed sequentially in several iterations using one copy of the $\CurR_{5,5,2}$ codes at a time.
	Red lines indicate the lines being recovered.

	\begin{figure}[ht]
		\centering{}
		\subfloat[Initial Pattern]{
			\centering{}
			\begin{tikzpicture}
				\GraphInit[vstyle=Classic]
				\SetGraphUnit{1}
				\SetUpVertex[Math, MinSize=1pt]

				\def\s{0.6}	

				\Vertex[x={\s*0}, y={ \s*0}, NoLabel]{a00}
				\Vertex[x={\s*1}, y={ \s*0}, NoLabel]{a10}
				\Vertex[x={\s*2}, y={ \s*0}, NoLabel]{a20}
				\Vertex[x={\s*3}, y={ \s*0}, NoLabel]{a30}
				\Vertex[x={\s*4}, y={ \s*0}, NoLabel]{a40}

				\Vertex[x={\s*0}, y={-\s*1}, NoLabel]{a01}
				\Vertex[x={\s*1}, y={-\s*1}, NoLabel]{a11}
				\Vertex[x={\s*2}, y={-\s*1}, NoLabel]{a21}
				\Vertex[x={\s*3}, y={-\s*1}, NoLabel]{a31}
				\Vertex[x={\s*4}, y={-\s*1}, NoLabel]{a41}

				\Vertex[x={\s*0}, y={-\s*2}, NoLabel]{a02}
				\Vertex[x={\s*1}, y={-\s*2}, NoLabel]{a12}
				\Vertex[x={\s*2}, y={-\s*2}, NoLabel]{a22}
				\Vertex[x={\s*3}, y={-\s*2}, NoLabel]{a32}
				\Vertex[x={\s*4}, y={-\s*2}, NoLabel]{a42}

				\Vertex[x={\s*0}, y={-\s*3}, NoLabel]{a03}
				\Vertex[x={\s*1}, y={-\s*3}, NoLabel]{a13}
				\Vertex[x={\s*2}, y={-\s*3}, NoLabel]{a23}
				\Vertex[x={\s*3}, y={-\s*3}, NoLabel]{a33}
				\Vertex[x={\s*4}, y={-\s*3}, NoLabel]{a43}

				\Vertex[x={\s*0}, y={-\s*4}, NoLabel]{a04}
				\Vertex[x={\s*1}, y={-\s*4}, NoLabel]{a14}
				\Vertex[x={\s*2}, y={-\s*4}, NoLabel]{a24}
				\Vertex[x={\s*3}, y={-\s*4}, NoLabel]{a34}
				\Vertex[x={\s*4}, y={-\s*4}, NoLabel]{a44}

				\AddVertexColor{black}{a00,a10,a21,a32}

				\tikzstyle{EdgeStyle}=[dotted]
				\foreach \y in {0,...,4}{%
					\Edges(a0\y, a1\y, a2\y, a3\y, a4\y)
				}
				\foreach \x in {0,...,4}{%
					\Edges(a\x0, a\x1, a\x2, a\x3, a\x4)
				}
				\end{tikzpicture}
		}
		\subfloat[First Iteration]{
			\centering{}
			\begin{tikzpicture}
				\GraphInit[vstyle=Classic]
				\SetGraphUnit{1}
				\SetUpVertex[Math, MinSize=1pt]

				\def\s{0.6}	

				\Vertex[x={\s*0}, y={ \s*0}, NoLabel]{a00}
				\Vertex[x={\s*1}, y={ \s*0}, NoLabel]{a10}
				\Vertex[x={\s*2}, y={ \s*0}, NoLabel]{a20}
				\Vertex[x={\s*3}, y={ \s*0}, NoLabel]{a30}
				\Vertex[x={\s*4}, y={ \s*0}, NoLabel]{a40}

				\Vertex[x={\s*0}, y={-\s*1}, NoLabel]{a01}
				\Vertex[x={\s*1}, y={-\s*1}, NoLabel]{a11}
				\Vertex[x={\s*2}, y={-\s*1}, NoLabel]{a21}
				\Vertex[x={\s*3}, y={-\s*1}, NoLabel]{a31}
				\Vertex[x={\s*4}, y={-\s*1}, NoLabel]{a41}

				\Vertex[x={\s*0}, y={-\s*2}, NoLabel]{a02}
				\Vertex[x={\s*1}, y={-\s*2}, NoLabel]{a12}
				\Vertex[x={\s*2}, y={-\s*2}, NoLabel]{a22}
				\Vertex[x={\s*3}, y={-\s*2}, NoLabel]{a32}
				\Vertex[x={\s*4}, y={-\s*2}, NoLabel]{a42}

				\Vertex[x={\s*0}, y={-\s*3}, NoLabel]{a03}
				\Vertex[x={\s*1}, y={-\s*3}, NoLabel]{a13}
				\Vertex[x={\s*2}, y={-\s*3}, NoLabel]{a23}
				\Vertex[x={\s*3}, y={-\s*3}, NoLabel]{a33}
				\Vertex[x={\s*4}, y={-\s*3}, NoLabel]{a43}

				\Vertex[x={\s*0}, y={-\s*4}, NoLabel]{a04}
				\Vertex[x={\s*1}, y={-\s*4}, NoLabel]{a14}
				\Vertex[x={\s*2}, y={-\s*4}, NoLabel]{a24}
				\Vertex[x={\s*3}, y={-\s*4}, NoLabel]{a34}
				\Vertex[x={\s*4}, y={-\s*4}, NoLabel]{a44}

				\AddVertexColor{black}{a00,a10,a21,a32}

				\tikzstyle{EdgeStyle}=[dotted]
				\foreach \y in {0,...,4}{%
					\Edges(a0\y, a1\y, a2\y, a3\y, a4\y)
				}
				\foreach \x in {0,...,4}{%
					\Edges(a\x0, a\x1, a\x2, a\x3, a\x4)
				}

				\tikzstyle{EdgeStyle}=[solid, color=red]
				\Edges(a00, a10, a20, a30, a40)
			\end{tikzpicture}
		}
		\subfloat[Second Iteration]{
			\centering{}
			\begin{tikzpicture}
				\GraphInit[vstyle=Classic]
				\SetGraphUnit{1}
				\SetUpVertex[Math, MinSize=1pt]

				\def\s{0.6}	

				\Vertex[x={\s*0}, y={ \s*0}, NoLabel]{a00}
				\Vertex[x={\s*1}, y={ \s*0}, NoLabel]{a10}
				\Vertex[x={\s*2}, y={ \s*0}, NoLabel]{a20}
				\Vertex[x={\s*3}, y={ \s*0}, NoLabel]{a30}
				\Vertex[x={\s*4}, y={ \s*0}, NoLabel]{a40}

				\Vertex[x={\s*0}, y={-\s*1}, NoLabel]{a01}
				\Vertex[x={\s*1}, y={-\s*1}, NoLabel]{a11}
				\Vertex[x={\s*2}, y={-\s*1}, NoLabel]{a21}
				\Vertex[x={\s*3}, y={-\s*1}, NoLabel]{a31}
				\Vertex[x={\s*4}, y={-\s*1}, NoLabel]{a41}

				\Vertex[x={\s*0}, y={-\s*2}, NoLabel]{a02}
				\Vertex[x={\s*1}, y={-\s*2}, NoLabel]{a12}
				\Vertex[x={\s*2}, y={-\s*2}, NoLabel]{a22}
				\Vertex[x={\s*3}, y={-\s*2}, NoLabel]{a32}
				\Vertex[x={\s*4}, y={-\s*2}, NoLabel]{a42}

				\Vertex[x={\s*0}, y={-\s*3}, NoLabel]{a03}
				\Vertex[x={\s*1}, y={-\s*3}, NoLabel]{a13}
				\Vertex[x={\s*2}, y={-\s*3}, NoLabel]{a23}
				\Vertex[x={\s*3}, y={-\s*3}, NoLabel]{a33}
				\Vertex[x={\s*4}, y={-\s*3}, NoLabel]{a43}

				\Vertex[x={\s*0}, y={-\s*4}, NoLabel]{a04}
				\Vertex[x={\s*1}, y={-\s*4}, NoLabel]{a14}
				\Vertex[x={\s*2}, y={-\s*4}, NoLabel]{a24}
				\Vertex[x={\s*3}, y={-\s*4}, NoLabel]{a34}
				\Vertex[x={\s*4}, y={-\s*4}, NoLabel]{a44}

				\AddVertexColor{black}{a00,a10,a20,a30,a40,a21,a32}

				\tikzstyle{EdgeStyle}=[dotted]
				\foreach \y in {0,...,4}{%
					\Edges(a0\y, a1\y, a2\y, a3\y, a4\y)
				}
				\foreach \x in {0,...,4}{%
					\Edges(a\x0, a\x1, a\x2, a\x3, a\x4)
				}

				\tikzstyle{EdgeStyle}=[solid, color=black]
				\Edges(a00, a10, a20, a30, a40)
				\tikzstyle{EdgeStyle}=[solid, color=red]
				\Edges(a20, a21, a22, a23, a24)
				\Edges(a30, a31, a32, a33, a34)
			\end{tikzpicture}
		}
		\subfloat[Third Iteration]{
			\centering{}
			\begin{tikzpicture}
				\GraphInit[vstyle=Classic]
				\SetGraphUnit{1}
				\SetUpVertex[Math, MinSize=1pt]

				\def\s{0.6}	

				\Vertex[x={\s*0}, y={ \s*0}, NoLabel]{a00}
				\Vertex[x={\s*1}, y={ \s*0}, NoLabel]{a10}
				\Vertex[x={\s*2}, y={ \s*0}, NoLabel]{a20}
				\Vertex[x={\s*3}, y={ \s*0}, NoLabel]{a30}
				\Vertex[x={\s*4}, y={ \s*0}, NoLabel]{a40}

				\Vertex[x={\s*0}, y={-\s*1}, NoLabel]{a01}
				\Vertex[x={\s*1}, y={-\s*1}, NoLabel]{a11}
				\Vertex[x={\s*2}, y={-\s*1}, NoLabel]{a21}
				\Vertex[x={\s*3}, y={-\s*1}, NoLabel]{a31}
				\Vertex[x={\s*4}, y={-\s*1}, NoLabel]{a41}

				\Vertex[x={\s*0}, y={-\s*2}, NoLabel]{a02}
				\Vertex[x={\s*1}, y={-\s*2}, NoLabel]{a12}
				\Vertex[x={\s*2}, y={-\s*2}, NoLabel]{a22}
				\Vertex[x={\s*3}, y={-\s*2}, NoLabel]{a32}
				\Vertex[x={\s*4}, y={-\s*2}, NoLabel]{a42}

				\Vertex[x={\s*0}, y={-\s*3}, NoLabel]{a03}
				\Vertex[x={\s*1}, y={-\s*3}, NoLabel]{a13}
				\Vertex[x={\s*2}, y={-\s*3}, NoLabel]{a23}
				\Vertex[x={\s*3}, y={-\s*3}, NoLabel]{a33}
				\Vertex[x={\s*4}, y={-\s*3}, NoLabel]{a43}

				\Vertex[x={\s*0}, y={-\s*4}, NoLabel]{a04}
				\Vertex[x={\s*1}, y={-\s*4}, NoLabel]{a14}
				\Vertex[x={\s*2}, y={-\s*4}, NoLabel]{a24}
				\Vertex[x={\s*3}, y={-\s*4}, NoLabel]{a34}
				\Vertex[x={\s*4}, y={-\s*4}, NoLabel]{a44}

				\AddVertexColor{black}{a00,a10,a20,a30,a40,a21,a22,a23,a24,a31,a32,a33,a34}

				\tikzstyle{EdgeStyle}=[dotted]
				\foreach \y in {0,...,4}{%
					\Edges(a0\y, a1\y, a2\y, a3\y, a4\y)
				}
				\foreach \x in {0,...,4}{%
					\Edges(a\x0, a\x1, a\x2, a\x3, a\x4)
				}

				\tikzstyle{EdgeStyle}=[solid, color=black]
				\Edges(a00, a10, a20, a30, a40)
				\Edges(a20, a21, a22, a23, a24)
				\Edges(a30, a31, a32, a33, a34)
				\tikzstyle{EdgeStyle}=[solid, color=red]
				\Edges(a01, a11, a21, a31, a41)
				\Edges(a02, a12, a22, a32, a42)
				\Edges(a03, a13, a23, a33, a43)
				\Edges(a04, a14, a24, a34, a44)
			\end{tikzpicture}
		}
		\caption{\centering{Iterative Recovery  Process}}\label{figure:iterations}
	\end{figure}
\end{example}


\subsubsection*{Erasure recovery in product codes}

For erasure correction, any erased symbol can be repaired using any of the codes in the product.
In general, consider an $(n_1 \cdots n_\ell, k_1 \cdots k_\ell, r, t)$-SLRC $\CurC = \CurC_1 \Cpro \cdots \Cpro \CurC_\ell$ where, for each $1 \leq j \leq \ell$, $\CurC_j$ is an $[n_j, k_j]$ code.
We can index the symbols of $\CurC$ by $(i_1, \ldots, i_\ell)$ where each $i_j \in [n_j]$.
Then, given some erased symbol $x$ at ${i_1, i_2, \ldots, i_\ell}$, for each $1 \leq j \leq \ell$ we can recover $x$ using any $r$-recovery set $R \subset \CurC_j$ where $i_j \in \supp(R)$
(see Example~\ref{ex:recovery} for an illustration).

\begin{lemma}\label{Lemma:maxpow}
	Let $a, c \in \R_+$ where $2 \leq a \leq c - 2$ and let $2 \leq \ell \in \N$.
	Then, for a fixed $c$, $a^\ell + {(c - a)}^\ell$ is maximized with $a = 2$ or with $a = c - 2$.
\end{lemma}

\begin{proof}
	Let $f(x) = x^\ell + {(c - x)}^\ell$.
	Then $f'(x) = \ell x^{\ell - 1} - \ell {(c - x)}^{\ell - 1}$ and $f'(c / 2) = 0$.
	For $2 \leq x < c / 2$, $f'(x) < 0$ and for $c / 2 < x \leq c - 2$, $f'(x) > 0$.
	Therefore $f(x)$ is maximized at $x = 2$ and $x = c - 2$ (and additionally $f(2) = f(c - 2)$).
\end{proof}

\begin{theorem}
    Let $\CurC$ be an ${[n, k, d]}_q$ code with $1\leq k < n$, let $\ell \geq 2$,
    and let $\mu$ denote the number of erasures.
	Then the $\ell$-fold product $\CurC^* := \CurC \Cpro \cdots \Cpro \CurC$ has the following erasure recovery behavior:
	\begin{enumerate}[label={(\alph*)}]
		\item If $\mu \leq \ell (d - 1)$, we can recover everything with only one copy of the ${[n, k, d]}_q$ code (in parallel).
		\item If $\ell (d - 1) < \mu < {d}^\ell$, we can fully sequentially recover, but we cannot necessarily recover in parallel.
		\item If ${d}^\ell \leq \mu \leq \min(n^\ell - k^\ell, \ell n^{\ell-1}(d-1))$, we can possibly recover depending on the erasure pattern. 
		\item If $\min(n^\ell - k^\ell, \ell n^{\ell-1}(d-1)) < \mu$, we definitely cannot fully recover.
	\end{enumerate}
\end{theorem}

\begin{proof}
	\begin{enumerate}[label={(\alph*)}]
		\item In this case there will be (at least) one projection from $\CurC^*$ onto the Cartesian product $\CurC^\ell$
			where each component has at most $d - 1$ erasures.
			Therefore, each of these components can recover their erasures in parallel.
		\item In this case, depending on the erasure pattern, parallel recovery is not always possible since all projections to $\CurC^\ell$ may contain at least one component with at least $d$ erasures.
			For example, consider a ${[5, 2, 4]}^2$ code whose $d^\ell - 1 = 15$ erasures are given in Figure~\ref{figure:recoverable-par}; we can recover each column (and therefore the full code) in parallel.
			For the lower bound, the same code with the pattern given by Figure~\ref{figure:recoverable-seq} has $\ell (d - 1) + 1 = 7$ erasures and is only sequentially recoverable.
		\item As an example, consider a ${[5, 2, 4]}_q$ code $\CurC$ and let $\CurC^* = \CurC \Cpro \CurC$.
			The erasure pattern given in Figure~\ref{figure:iterations}, has $\mu = n^\ell - {(n - d + 1)}^\ell = 21$ erasures but is fully (sequentially) recoverable.
			However, the erasure pattern given in Figure~\ref{figure:unrecoverable} having 16 erasures cannot be recovered.
			Additionally, since $k + d \leq n + 1$, $n \geq 3$, and $\ell \geq 2$, by Lemma~\ref{Lemma:maxpow} we have that
			\begin{align*}
				k^\ell + d^\ell
				&\leq 2^\ell + {(n - 1)}^\ell
				= 2^\ell + n^\ell - \sum_{i = 0}^{\ell - 1} \binom{\ell}{i} {(n - 1)}^i \\
				&\leq 2^\ell + n^\ell - \sum_{i = 0}^{\ell - 1} \binom{\ell}{i} 2^i \\
				&\leq 2^\ell + n^\ell - \ell 2^{\ell - 1} \\
				&\leq n^\ell.
			\end{align*}
			Similarly, we can also prove that $d^\ell \leq \ell n^{\ell-1}(d-1)$ for $d>1$. For this note that since $d,n,\ell\geq 2$ we have 
			$$\frac{d^\ell}{d-1} \leq \frac{n^\ell}{n-1}\iff d^\ell\leq \frac{n^\ell}{n-1}(d-1)$$
			and the latter is upper bounded by $\ell n^{\ell-1} (d-1)$ because $\ell \geq 2 \geq n/(n-1)$.
		\item First, if we have less than $n^\ell-(n^\ell-k^\ell)=k^\ell$ coordinates left, those un-erased coordinates cannot uniquely determine a codeword of the $k^\ell$-dimensional code $\CurC^*$. On the other hand, the maximal number of node recoveries we could do is $d-1$ per line, i.e., (since there are $\ell n^{\ell-1}$ axis-parallel lines) at most
			$$ \ell n^{\ell-1}(d-1)$$
			recoveries.
	\end{enumerate}
\end{proof}

\begin{figure}[ht]
    \centering{}
    \begin{tikzpicture}
        \GraphInit[vstyle=Classic]
        \SetGraphUnit{1}
        \SetUpVertex[Math, MinSize=1pt]

        \def\s{0.6}	

        \Vertex[x={\s*0}, y={ \s*0}, NoLabel]{a00}
        \Vertex[x={\s*1}, y={ \s*0}, NoLabel]{a10}
        \Vertex[x={\s*2}, y={ \s*0}, NoLabel]{a20}
        \Vertex[x={\s*3}, y={ \s*0}, NoLabel]{a30}
        \Vertex[x={\s*4}, y={ \s*0}, NoLabel]{a40}

        \Vertex[x={\s*0}, y={-\s*1}, NoLabel]{a01}
        \Vertex[x={\s*1}, y={-\s*1}, NoLabel]{a11}
        \Vertex[x={\s*2}, y={-\s*1}, NoLabel]{a21}
        \Vertex[x={\s*3}, y={-\s*1}, NoLabel]{a31}
        \Vertex[x={\s*4}, y={-\s*1}, NoLabel]{a41}

        \Vertex[x={\s*0}, y={-\s*2}, NoLabel]{a02}
        \Vertex[x={\s*1}, y={-\s*2}, NoLabel]{a12}
        \Vertex[x={\s*2}, y={-\s*2}, NoLabel]{a22}
        \Vertex[x={\s*3}, y={-\s*2}, NoLabel]{a32}
        \Vertex[x={\s*4}, y={-\s*2}, NoLabel]{a42}

        \Vertex[x={\s*0}, y={-\s*3}, NoLabel]{a03}
        \Vertex[x={\s*1}, y={-\s*3}, NoLabel]{a13}
        \Vertex[x={\s*2}, y={-\s*3}, NoLabel]{a23}
        \Vertex[x={\s*3}, y={-\s*3}, NoLabel]{a33}
        \Vertex[x={\s*4}, y={-\s*3}, NoLabel]{a43}

        \Vertex[x={\s*0}, y={-\s*4}, NoLabel]{a04}
        \Vertex[x={\s*1}, y={-\s*4}, NoLabel]{a14}
        \Vertex[x={\s*2}, y={-\s*4}, NoLabel]{a24}
        \Vertex[x={\s*3}, y={-\s*4}, NoLabel]{a34}
        \Vertex[x={\s*4}, y={-\s*4}, NoLabel]{a44}

        \AddVertexColor{black}{a00,a10,a20,a30,a40,a01,a11,a21,a31,a41}

        \tikzstyle{EdgeStyle}=[dotted]
        \foreach \y in {0,...,4}{%
            \Edges(a0\y, a1\y, a2\y, a3\y, a4\y)
        }
        \foreach \x in {0,...,4}{%
            \Edges(a\x0, a\x1, a\x2, a\x3, a\x4)
        }
    \end{tikzpicture}
    \caption{\centering{Parallel recoverable erasure pattern}}\label{figure:recoverable-par}
\end{figure}

\begin{figure}[ht]
    \centering{}
    \begin{tikzpicture}
        \GraphInit[vstyle=Classic]
        \SetGraphUnit{1}
        \SetUpVertex[Math, MinSize=1pt]

        \def\s{0.6}	

        \Vertex[x={\s*0}, y={ \s*0}, NoLabel]{a00}
        \Vertex[x={\s*1}, y={ \s*0}, NoLabel]{a10}
        \Vertex[x={\s*2}, y={ \s*0}, NoLabel]{a20}
        \Vertex[x={\s*3}, y={ \s*0}, NoLabel]{a30}
        \Vertex[x={\s*4}, y={ \s*0}, NoLabel]{a40}

        \Vertex[x={\s*0}, y={-\s*1}, NoLabel]{a01}
        \Vertex[x={\s*1}, y={-\s*1}, NoLabel]{a11}
        \Vertex[x={\s*2}, y={-\s*1}, NoLabel]{a21}
        \Vertex[x={\s*3}, y={-\s*1}, NoLabel]{a31}
        \Vertex[x={\s*4}, y={-\s*1}, NoLabel]{a41}

        \Vertex[x={\s*0}, y={-\s*2}, NoLabel]{a02}
        \Vertex[x={\s*1}, y={-\s*2}, NoLabel]{a12}
        \Vertex[x={\s*2}, y={-\s*2}, NoLabel]{a22}
        \Vertex[x={\s*3}, y={-\s*2}, NoLabel]{a32}
        \Vertex[x={\s*4}, y={-\s*2}, NoLabel]{a42}

        \Vertex[x={\s*0}, y={-\s*3}, NoLabel]{a03}
        \Vertex[x={\s*1}, y={-\s*3}, NoLabel]{a13}
        \Vertex[x={\s*2}, y={-\s*3}, NoLabel]{a23}
        \Vertex[x={\s*3}, y={-\s*3}, NoLabel]{a33}
        \Vertex[x={\s*4}, y={-\s*3}, NoLabel]{a43}

        \Vertex[x={\s*0}, y={-\s*4}, NoLabel]{a04}
        \Vertex[x={\s*1}, y={-\s*4}, NoLabel]{a14}
        \Vertex[x={\s*2}, y={-\s*4}, NoLabel]{a24}
        \Vertex[x={\s*3}, y={-\s*4}, NoLabel]{a34}
        \Vertex[x={\s*4}, y={-\s*4}, NoLabel]{a44}

        \AddVertexColor{black}{a00,a10,a20,a30,a40,a01,a11,a21,a31,a02,a12,a22,a32,a03,a13,a23,a33,a04}

        \tikzstyle{EdgeStyle}=[dotted]
        \foreach \y in {0,...,4}{%
            \Edges(a0\y, a1\y, a2\y, a3\y, a4\y)
        }
        \foreach \x in {0,...,4}{%
            \Edges(a\x0, a\x1, a\x2, a\x3, a\x4)
        }
    \end{tikzpicture}
    \caption{\centering{Sequential recoverable erasure pattern}}\label{figure:recoverable-seq}
\end{figure}
\begin{figure}[ht]
    \centering{}
    \begin{tikzpicture}
        \GraphInit[vstyle=Classic]
        \SetGraphUnit{1}
        \SetUpVertex[Math, MinSize=1pt]

        \def\s{0.6}	

        \Vertex[x={\s*0}, y={ \s*0}, NoLabel]{a00}
        \Vertex[x={\s*1}, y={ \s*0}, NoLabel]{a10}
        \Vertex[x={\s*2}, y={ \s*0}, NoLabel]{a20}
        \Vertex[x={\s*3}, y={ \s*0}, NoLabel]{a30}
        \Vertex[x={\s*4}, y={ \s*0}, NoLabel]{a40}

        \Vertex[x={\s*0}, y={-\s*1}, NoLabel]{a01}
        \Vertex[x={\s*1}, y={-\s*1}, NoLabel]{a11}
        \Vertex[x={\s*2}, y={-\s*1}, NoLabel]{a21}
        \Vertex[x={\s*3}, y={-\s*1}, NoLabel]{a31}
        \Vertex[x={\s*4}, y={-\s*1}, NoLabel]{a41}

        \Vertex[x={\s*0}, y={-\s*2}, NoLabel]{a02}
        \Vertex[x={\s*1}, y={-\s*2}, NoLabel]{a12}
        \Vertex[x={\s*2}, y={-\s*2}, NoLabel]{a22}
        \Vertex[x={\s*3}, y={-\s*2}, NoLabel]{a32}
        \Vertex[x={\s*4}, y={-\s*2}, NoLabel]{a42}

        \Vertex[x={\s*0}, y={-\s*3}, NoLabel]{a03}
        \Vertex[x={\s*1}, y={-\s*3}, NoLabel]{a13}
        \Vertex[x={\s*2}, y={-\s*3}, NoLabel]{a23}
        \Vertex[x={\s*3}, y={-\s*3}, NoLabel]{a33}
        \Vertex[x={\s*4}, y={-\s*3}, NoLabel]{a43}

        \Vertex[x={\s*0}, y={-\s*4}, NoLabel]{a04}
        \Vertex[x={\s*1}, y={-\s*4}, NoLabel]{a14}
        \Vertex[x={\s*2}, y={-\s*4}, NoLabel]{a24}
        \Vertex[x={\s*3}, y={-\s*4}, NoLabel]{a34}
        \Vertex[x={\s*4}, y={-\s*4}, NoLabel]{a44}

        \AddVertexColor{black}{a00,a10,a20,a30,a40,a01,a02,a03,a04}

        \tikzstyle{EdgeStyle}=[dotted]
        \foreach \y in {0,...,4}{%
            \Edges(a0\y, a1\y, a2\y, a3\y, a4\y)
        }
        \foreach \x in {0,...,4}{%
            \Edges(a\x0, a\x1, a\x2, a\x3, a\x4)
        }
    \end{tikzpicture}
    \caption{\centering{Unrecoverable erasure pattern}}\label{figure:unrecoverable}
\end{figure}


\section{Comparison to other LRCs}

In the general literature on LRCs, most work focuses on either parallel recovery, binary alphabets, or both.
Additionally, much of the work focuses on codes with a small number of erasures (e.g., $t = 2$). In contrast, our construction is very general and works for any field size $q$ and number of erasures $t$. 


\subsubsection*{Parallel recoverable codes}

We note that optimal constructions for LRCs exist, in particular for large $n$ relative to $r$ and $d$ (for example, see~\cite{GXY19}). In particular, whenever
\begin{align*}
	\frac{n}{r + 1} \geq \left(d - 2 - \left\lfloor \frac{d - 2}{r + 1} \right\rfloor \right) (3r + 2) + \left\lfloor \frac{d - 2}{r + 1} \right\rfloor + 1
\end{align*}
then without loss of generality one can assume that optimal LRCs have disjoint recovery sets and can hence recover all recoverable erasure patterns with parallel recovery~\cite{GXY19}. 
In the case of the parameters given previously using our construction, this bound is never met, meaning that sequential recovery can be better than parallel recovery.

Additionally, Tamo and Barg~\cite{TB14} gave optimal constructions using evaluations of polynomials for the case when $q \geq n$ and when there exists a ``good'' polynomial for the desired value of $r$. The results were later extended to codes with high availability—an analog of alternativity in the parallel recovery setting—in~\cite{GMNLW23} for some parameter sets using codes from evaluations of polynomials on paradoxical families of subsets.

Therefore, we focused on constructing codes filling in the (many) gaps outside of these categories, i.e., $q$-ary (for $2<q<n$), sequentially recoverable codes with large $d$ (or equivalently $t$) and small $n$ relative to $r$.


\subsubsection*{SLRCs derived from graph-based constructions}

In~\cite{PLBK19}, (in their words, ``rate-'' and ``distance-'') optimal SLRCs are constructed from both Tur\'{a}n and regular graphs.
These specifically target the $t = 2$ case and gives codes with rate $\tfrac{r}{r + 2}$.

Similarly, using Tur\'{a}n hypergraphs,~\cite{BPK16-1} extended these results to the $t = 3$ case with rate $\tfrac{r}{r + 3}$.

In~\cite{BKK19}, optimal binary SLRCs are built from Moore graphs---regular graphs whose girth is more than twice its diameter.
However, these corresponding graphs only exist for $t \in \{2, 3, 4, 5, 7, 11\}$ when $r \geq 2$.

Compared to them our construction is more general and works for any $q$ and $t$.


\subsubsection*{Previously known SLRCs derived from product codes}

In~\cite{SY16}, products of ${[n, n - 1, 2]}_2$ codes are presented; our method generalizes this to other codes over $q$-ary fields.

Following this, the authors of~\cite{SCYCH18} derived a new family of specific instances of binary product codes which give a rate of
\begin{align*}
	{\left( \sum_{s = 0}^{\ell} r^{-\abs{\supp_{\ell}(s)}} \right)}^{-1},
\end{align*}
where $r \geq 2$, $t \geq 1$, $\ell$ is any positive integer satisfying $t \leq 2^{\ell} - 1$, and $\supp_{\ell}(s)$ is the support of the $\ell$-digit binary representation of $s$.
However, this construction provides only sporadic parameter sets $(n, k, r, t)$; in particular, for $k \leq 10$, we have parameter sets for $k = 4$ and $r = 2$, $k = 8$ and $r = 2$, and $k = 9$ and $r = 3$ where in each case $t < k$.


\subsubsection*{Partial MDS codes}

We can also consider our codes in comparison to partial MDS (PMDS) codes. 
As mentioned in Section~\ref{sec:prelim} these PMDS codes are information-theoretically optimal LRCs but the full number of erasures cannot be recovered locally; i.e., in each block $i$, $t_i$ erasures can be recovered with locality $r$ but the remaining $s = mr - k$ erasures must use a full-rank minor of the generator matrix, i.e., a recovery set of size $k$. Our construction, on the other hand, provides locality for any erasure pattern.

Additionally, for most constructions, PMDS codes require a field size exponential in the number $s$ of global erasures and hence yield constructions over much larger fields than this work provides.


\subsection*{Rate comparison to binary SLRCs}

In the following table we exemplify achievable rates of our construction compared to the binary constructions of~\cite{WZL15} and~\cite{SCYCH18}. The rates of our construction are taken from the examples in Section~\ref{sec:examples}, for either $q=3$ or $q=5$. As one can see, our construction achieves a better (or equal) rate than the other two constructions, in particular, for large  values of $t$.

\begin{table}[ht]
	\centering{}
	\begin{tabular}{cccccccccc}
		\toprule{}
		rate $\backslash$ $t$ & 2 & 3 & 4 & 5 & 6 & 7 & 8 & 9 & 10 \\
		\midrule{}
		this work & 0.5 & 0.44 & 0.39 & 0.33 & 0.32 & 0.30 & 0.25 & 0.26 & 0.22 \\
		\midrule{}
		$\frac{r}{r + t}$\cite{WZL15} & 0.5 & 0.4 & 0.33 & 0.29 & 0.28 & 0.22 & 0.2 & 0.18 & 0.17 \\
		\midrule{}
		${(\frac{r}{r + 1})}^t$\cite{SCYCH18} & 0.44 & 0.30 & 0.20 & 0.13 & 0.09 & 0.06 & 0.04 & 0.03 & 0.02 \\
		\bottomrule{}
	\end{tabular}
	\caption{\centering Achievable rates for locality $r = 2$ and erasure recovery capacity $t \in \{1, 2, \ldots, 10\}$.}\label{tab:CompTable}
\end{table}


\section{Conclusion}
In this paper, we used the Kronecker product to derive a general code construction for  $q$-ary  sequential locally recoverable codes (SLRCs) which are capable of recovering a general number of $t \geq 2$ erasures, with small locality $r$. 
We derived new bounds on the maximum number of recoverable erasures and a minimum number of repair alternatively.
Our construction, when using BCH and/or MDS codes as the building blocks, has a code rate greater than  previously known constructions with the same property.

Furthermore, the structure in our codes coming from the code product can be used for recovering erasures in an algorithmic way.  

In future work we would like to explore more about the connection between SLRCs and  geometric objects or block designs, to construct more sequential locally recoverable codes with new parameters.


\printbibliography[heading=bibnumbered]{}


\begin{appendices} 

\section{Geometric Proof of Corollary~\ref{cor:prod}}\label{app:geoprod}

\begin{proof}
	The length and dimension of the code follow the code's product definition. It remains to show the locality, the number of recoverable erasures, and the alternativity.
	
	We proceed by induction on $\ell$.
	For $\ell = 1$ this is trivial.
	Assume true up to $\ell - 1$ and let $\CurC = {\CurC}_1 \Cpro \cdots \Cpro {\CurC}_{\ell - 1}$, by our assumption, be an $(N_{\ell - 1}, K_{\ell - 1}, r, T_{\ell - 1})$-SLRC\@.
	We then consider the product $\CurC \Cpro {\CurC}_{\ell}$.
	
	Suppose $A$, $B$, and $C = A \otimes B$ respectively are generator matrices of $\CurC$, ${\CurC}_{\ell}$, and $\CurC \Cpro {\CurC}_{\ell}$ where $a_{ij}$ and $b_{ij}$ respectively represent the entries of $A$ and $B$.
	Let
	\begin{align*}
		m = (m_{11}, m_{12}, \ldots, m_{1k_\ell}, m_{21}, \ldots, m_{K_{\ell - 1}k_\ell}) \in \mathbb{F}_q^{K_\ell}
	\end{align*}
	and $z = mC \in \CurC \Cpro {\CurC}_{\ell}$ be the corresponding codeword.
	We represent $z$ as
	\begin{align*}
		z = (z_{11}, z_{12}, \ldots, z_{1n_\ell}, z_{21}, \ldots, z_{N_{\ell - 1}n_\ell}) \in \mathbb{F}_q^{N_\ell}
	\end{align*}
	and, by expanding the product $m (A \otimes B)$, we can see that
	\begin{align*}
		z_{ij} = \sum_{g = 1}^{K_{\ell - 1}} \sum_{h = 1}^{k_\ell} m_{gh} a_{gi} b_{hj}.
	\end{align*}
	Finally, let $m_j = (m_{1j}, \ldots, m_{K_{\ell - 1}})$, $x_j = m_j A$, and let ${(x_j)}_i$ represent the $i^\text{th}$ symbol of $x_j$.
	
	We now show that for each symbol $z_{ij}$ of $z$ and for each possible recovery set for a symbol in $\CurC$, there is a corresponding recovery set in the product code.
	Let $R$ be a $(r, \CurC)$-recovery set for $i$ with corresponding coefficients $c_s$ for $s \in R$.
	Then
	\begin{align*}
		z_{ij}
		&= \sum_{h = 1}^{k_\ell} b_{hj} \sum_{g = 1}^{K_{\ell - 1}} m_{gh} a_{gi}.
	\end{align*}
	Note that
	\begin{align*}
		\sum_{g = 1}^{K_{\ell - 1}} m_{gh} a_{gi}
		= {(x_h)}_i
		= \sum_{s \in R} c_s {(x_h)}_s.
	\end{align*}
	It then follows that
	\begin{align*}
		z_{ij}
		&= \sum_{h = 1}^{k_\ell} b_{hj} \sum_{s \in R} c_s {(x_h)}_s \\
		&= \sum_{s \in R} c_s \sum_{h = 1}^{k_\ell} b_{hj} {(x_h)}_s \\
		&= \sum_{s \in R} c_s \sum_{h = 1}^{k_\ell} b_{hj} \sum_{g = 1}^{K_{\ell - 1}} m_{gh} a_{gs} \\
		&= \sum_{s \in R} c_s z_{sj}.
	\end{align*}
	
	A similar argument holds using recovery sets coming from ${\CurC}_{\ell}$.
	Thus, we can see that the alternativity of the product is at least the sum of the alternativities of $\CurC$ and $\CurC_\ell$.

	To calculate the total number of recoverable erasures, we consider the product $\CurC \Cpro {\CurC}_{\ell}$ as consisting of $n_{\ell}$ copies of $\CurC$.
	If we have
	\begin{align*}
		T_{\ell}
		= (T_{\ell - 1} + 1)(t_{\ell} + 1) - 1
		= t_{\ell} (T_{\ell - 1} + 1) + T_{\ell - 1}
	\end{align*}
	erasures, under the worst case, we observe $T_{\ell - 1} + 1$ failed nodes on $t_{\ell}$ copies of $\CurC$ making these locally unrecoverable.
	However, this leaves only $T_{\ell - 1}$ erasures meaning the other $n_{\ell} - t_{\ell}$ copies of $\CurC$ can be fully recovered.
	Since all remaining unrecovered erasures lie in a copy of ${\CurC}_{\ell}$ with at most $t_{\ell}$ erasures, these also can be recovered.
	For an illustration of this see Figure~\ref{figure:product-recovery}. 
	
	\begin{figure}[ht]
		\begin{center}
			\begin{tikzpicture}[scale=0.75]
				\GraphInit[vstyle=Classic]
				\SetGraphUnit{1}
				\SetUpVertex[Math, MinSize=1pt]
	
				\def\s{1.5}	
				\def\xs{1.5}   
				\def\zr{0.3} 
				\def\za{30}  
	
				\foreach \x in {0,...,3}
				\foreach \y in {0,...,3}
				\foreach \z in {0,...,3}{%
					\Vertex[%
						x={\s*\xs*\x + cos{\za}*\s*\zr*\z},
						y={\s*\y + sin{\za}*\s*\zr*\z},
						NoLabel
					]{a\x\y\z}
				}
	
				\tikzstyle{EdgeStyle}=[solid, color=green]
				\foreach \x in {0, 1}{%
					\Edges(a\x00, a\x10, a\x20, a\x30)
					\Edges(a\x01, a\x11, a\x21, a\x31)
					\Edges(a\x02, a\x12, a\x22, a\x32)
					\Edges(a\x03, a\x13, a\x23, a\x33)
					\Edges(a\x00, a\x01, a\x02, a\x03)
					\Edges(a\x10, a\x11, a\x12, a\x13)
					\Edges(a\x20, a\x21, a\x22, a\x23)
					\Edges(a\x30, a\x31, a\x32, a\x33)
				}
	
				\tikzstyle{EdgeStyle}=[solid, color=red]
				\foreach \x in {2, 3}{%
					\Edges(a\x00, a\x10, a\x20, a\x30)
					\Edges(a\x01, a\x11, a\x21, a\x31)
					\Edges(a\x02, a\x12, a\x22, a\x32)
					\Edges(a\x03, a\x13, a\x23, a\x33)
					\Edges(a\x00, a\x01, a\x02, a\x03)
					\Edges(a\x10, a\x11, a\x12, a\x13)
					\Edges(a\x20, a\x21, a\x22, a\x23)
					\Edges(a\x30, a\x31, a\x32, a\x33)
				}
	
				\tikzstyle{EdgeStyle}=[dotted, color=black]
				\foreach \y in {0,1,2,3}
				\foreach \z in {0,1,2,3}{%
					\Edges(a0\y\z, a1\y\z, a2\y\z, a3\y\z)
				}
	
				\foreach \x in {0,1,2,3}{%
					\node(b\x0) at (\s*\xs*\x, 0){};
					\node(b\x1) at (\s*\xs*\x + cos{\za}*\s*\zr*3, 0){};
				}
	
				\draw[%
					thick,
					decoration={%
						brace,
						raise=0.5cm
					},
					decorate
				](a033) -- (a333);
				\node[above of=a033,xshift=3.5cm,align=center]{%
					Recovery sets of $ {\CurC_\ell}$
				};
				\draw[%
					thick,
					decoration={%
						brace,
						mirror,
						raise=0.5cm
					},
					decorate
				](b00) -- (b11);
				\node[below of=b00,xshift=1.75cm,align=center]{%
					Locally recoverable\\copies of $\CurC$
				};
	
				\draw[%
					thick,
					decoration={%
						brace,
						mirror,
						raise=0.5cm
					},
					decorate
				](b20) -- (b31);
				\node[below of=b20,xshift=1.75cm,align=center]{%
					Locally unrecoverable\\copies of $\CurC$
				}; 
			\end{tikzpicture}
		\end{center}
		\caption{\centering{Sequential recovery with the product code}}\label{figure:product-recovery}
	\end{figure}
\end{proof}

\end{appendices}


\end{document}